\documentclass[11pt, draftcls, onecolumn]{IEEEtran}
\usepackage{amsmath,epsfig}
\usepackage{cite}
\usepackage{graphicx,subfigure}
\usepackage{amssymb}
\usepackage{multirow}
\begin{document}

\newtheorem{lem}{Lemma}
\newtheorem{prop}{Proposition}
\newtheorem{cor}{Corollary}
\newtheorem{remark}{Remark}
\newtheorem{defin}{Definition}
\newtheorem{thm}{Theorem}

\newcounter{MYtempeqncnt}

\title{Joint Wireless Information and Energy Transfer in a Two-User MIMO Interference Channel}

\author{Jaehyun Park,~\IEEEmembership{Member,~IEEE,}
     Bruno Clerckx,~\IEEEmembership{Member,~IEEE,}
\thanks{J. Park is with the Broadcasting and
Telecommunications Convergence Research Laboratory, Electronics
and Telecommunications Research Institute (ETRI), Daejeon, Korea
(e-mail: jhpark@ee.kaist.ac.kr)}
\thanks{B. Clerckx is with the Department of Electrical and Electronic
Engineering, Imperial College London, South Kensington Campus
London SW7 2AZ, United Kingdom (e-mail:b.clerckx@imperial.ac.uk)}
}

\maketitle

\begin{abstract}
This paper investigates joint wireless information and energy
transfer in a two-user MIMO interference channel, in which each
receiver either decodes the incoming information data (information
decoding, ID) or harvests the RF energy (energy harvesting, EH) to
operate with a potentially perpetual energy supply. In the
two-user interference channel, we have four different scenarios
according to the receiver mode -- ($ID_1$, $ID_2$), ($EH_1$,
$EH_2$), ($EH_1$, $ID_2$), and ($ID_1$, $EH_2$). While the maximum
information bit rate is unknown and finding the optimal
transmission strategy is still open for ($ID_1$, $ID_2$), we have
derived the optimal transmission strategy achieving the maximum
harvested energy for ($EH_1$, $EH_2$). For ($EH_1$, $ID_2$), and
($ID_1$, $EH_2$), we find a necessary condition of the optimal
transmission strategy and, accordingly, identify the achievable
rate-energy (R-E) tradeoff region for two transmission strategies
that satisfy the necessary condition - maximum energy beamforming
(MEB) and minimum leakage beamforming (MLB). Furthermore, a new
transmission strategy satisfying the necessary condition -
signal-to-leakage-and-energy ratio (SLER) maximization beamforming
- is proposed and shown to exhibit a better R-E region than the
MEB and the MLB strategies. Finally, we propose a mode scheduling
method to switch between ($EH_1$, $ID_2$) and ($ID_1$, $EH_2$)
based on the SLER.
\end{abstract}


\section{Introduction}
\label{sec:intro} Over the last decade, there has been a lot of
interest to transfer energy wirelessly and recently,
radio-frequency (RF) radiation has become a viable source for
energy harvesting. It is nowadays possible to transfer the energy
wirelessly with a reasonable efficiency over small distances and,
furthermore, the wireless sensor network (WSN) in which the
sensors are capable of harvesting RF energy to power their own
transmissions has been introduced in industry (\cite{Soljacic,
Yates, Vullers, Fiez} and references therein).

The energy harvesting function can be exploited in either transmit
side \cite{Ozel,RRajesh, RRajesh1, KIshibashi, YLuo} or receive
side \cite{Zhang1, Zhang2, KHuang1, ANasir}. For the energy
harvesting transmitter, energy harvesting scheduling and transmit
power allocation have been considered and, for the energy
harvesting receiver, the management of information decoding and
energy harvesting has been developed. Furthermore, because RF
signals carry information as well as energy, ``joint wireless
information and energy transfer'' in conjunction with the energy
harvesting receiver has been investigated \cite{Zhang1, Zhang2,
KHuang1, ANasir}. That is, previous works have studied the
fundamental performance limits and the optimal transmission
strategies of the joint wireless information and energy transfer
in the cellular downlink system with a single base station (BS)
and multiple mobile stations (MSs) \cite{KHuang1} and in the
cooperative relay system \cite{ANasir} and in the broadcasting
system \cite{Zhang1, Zhang2} with a single energy receiver and a
single information receiver when they are separately located or
co-located.


There have been very few studies of joint wireless information and
energy transfer on the interference channel (IFC) models
\cite{Tutuncuoglu1, Tutuncuoglu2, KHuang2}. In \cite{Tutuncuoglu1,
Tutuncuoglu2}, the authors have considered a two-user single-input
single-output (SISO) IFC and derived the optimal power scheduling
at the energy harvesting transmitters that maximizes the sum-rate
given harvested energy constraints. In \cite{KHuang2}, the authors
have investigated joint information and energy transfer in
multi-cell cellular networks with single-antenna BSs and
single-antenna MSs. To the best of the authors' knowledge, the
general setup of multiple-input multiple-output (MIMO) IFC models
accounting for joint wireless information and energy transfer has
not been addressed so far.

As an initial step, in this paper, we investigate a joint wireless
information and energy transfer in a two-user MIMO IFC, where each
receiver either decodes the incoming information data (information
decoding, ID) or harvests the RF energy (energy harvesting, EH) to
operate with a potentially perpetual energy supply. Because
practical circuits and hardware that harvest energy from the
received RF signal are not yet able to decode the information
carried through the same RF signal \cite{Zhang1, Zhang2,
ZhouZhangHo}, we assume that the receiver cannot decode the
information and simultaneously harvest
energy. 
It is also assumed that the two (Tx 1,Tx 2) transmitters have
knowledge of their local CSI only, i.e. the CSI corresponding to
the links between a transmitter and all receivers (Rx 1, Rx 2). In
addition, the transmitters do not share the information data to be
transmitted and their CSI and, furthermore, the interference is
assumed not decodable at the receiver nodes as in \cite{Shen}.
That is, Tx 1 (Tx 2) cannot transfer the information to Rx 2 (Rx
1). In a two-user IFC, we then have four different scenarios
according to the Rx mode -- ($ID_1$, $ID_2$), ($EH_1$, $EH_2$),
($EH_1$, $ID_2$), and ($ID_1$, $EH_2$). Because, for ($ID_1$,
$ID_2$), the maximum information bit rate is unknown and finding
the optimal transmission strategy is still an open problem in
general, we investigate the achievable rate when a well-known
iterative water-filling algorithm \cite{Scutari,WeiYu,AGoldsmith}
is adopted for ($ID_1$, $ID_2$) with no CSI sharing between two
transmitters. For ($EH_1$, $EH_2$), we derive the optimal
transmission strategy achieving the maximum harvested energy.
Because the receivers operate in a single mode such as ($ID_1$,
$ID_2$) and ($EH_1$, $EH_2$), when the information is transferred,
no energy is harvested from RF signals and vice versa. For
($EH_1$, $ID_2$) and ($ID_1$, $EH_2$), the achievable energy-rate
(R-E) trade-off region is not easily identified and the optimal
transmission strategy is still unknown. However, in this paper, we
find a necessary condition of the optimal transmission strategy,
in which one of the transmitters should take a rank-one energy
beamforming strategy with a proper power control. Accordingly, the
achievable R-E tradeoff region is identified for two different
rank-one beamforming strategies - maximum energy beamforming (MEB)
and minimum leakage beamforming (MLB). 
Furthermore, we also propose a new transmission strategy that
satisfies the necessary condition - signal-to-leakage-and-energy
ratio (SLER) maximization beamforming. Note that the SLER
maximizing approach is comparable to the
signal-to-leakage-and-noise ratio (SLNR) maximization beamforming
\cite{Sadek, JPark} which has been developed for the multi-user
MIMO data transmission, not considering the energy transfer. The
simulation results demonstrate that the proposed SLER maximization
strategy exhibits wider R-E region than the conventional
transmission methods such as MLB, MEB, and SLNR beamforming.
Finally, we propose a mode scheduling method to switch between
($EH_1$, $ID_2$) and ($ID_1$, $EH_2$) based on the SLER that
further extends R-E tradeoff region.

The rest of this paper is organized as follows. In Section
\ref{sec:systemmodel}, we introduce the system model for two-user
MIMO IFC. In Section \ref{sec:single}, we discuss the transmission
strategy for two receivers on a single mode, i.e. ($ID_1$, $ID_2$)
and ($EH_1$, $EH_2$). In Section \ref{sec:oneIDoneEH}, we derive
the necessary condition for the optimal transmission strategy and
investigate the achievable rate-energy (R-E) region for ($EH_1$,
$ID_2$) and ($ID_1$, $EH_2$) and, in Section
\ref{sec:SLER_maximizing}, propose the SLER maximization strategy.
In Section \ref{sec:discussion} and Section \ref{sec:simulation},
we provide several discussion and simulation results,
respectively, and in Section \ref{sec:conc} we give our
conclusions.

Throughout the paper, matrices and vectors are represented by bold
capital letters and bold lower-case letters, respectively. The
notations $({\bf A})^{H} $, $({\bf A})^{\dagger} $, $({\bf A})_i$,
$[{\bf A}]_i$, $tr({\bf A})$, and $\det({\bf A})$ denote the
conjugate transpose, pseudo-inverse, the $i$th row, the $i$th
column, the trace, and the determinant of a matrix ${\bf A}$,
respectively. The matrix norm $\|{\bf A}\|$ and $\|{\bf A}\|_F$
denote the 2-norm and Frobenius norm of a matrix ${\bf A}$,
respectively, and the vector norm $\|{\bf a}\|$ denotes the 2-norm
of a vector ${\bf a}$. In addition, $(a)^+ \triangleq \max (a, 0)$
and ${\bf A} \succeq 0$ means that a matrix ${\bf A}$ is positive
semi-definite. Finally, ${\bf I}_{M}$ denotes the $M \times M$
identity matrix.

\section{System model}
\label{sec:systemmodel}

\begin{figure}
\begin{center}
\begin{tabular}{c}
\includegraphics[height=4cm]{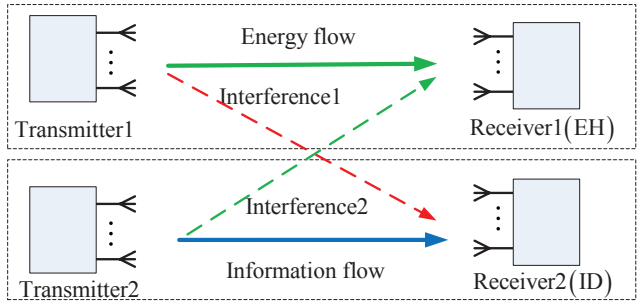}
\end{tabular}
\end{center}
\caption[JWIET_two_userIC_block]
{ \label{JWIET_two_userIC_block} Two-user MIMO IFC in ($EH_1$,
$ID_2$) mode.}
\end{figure}

We consider a two-user MIMO IFC system where two transmitters,
each with $M_t$ antennas, are simultaneously transmitting their
signals to two receivers, each with $M_r$ antennas, as shown in
Fig. \ref{JWIET_two_userIC_block}. Note that each receiver can
either decode the information or harvest energy from the received
signal, but it cannot execute the information decoding and energy
harvesting at the same time due to the hardware limitations. That
is, each receiver can switch between ID
mode and EH mode at each frame or time slot.\footnote{
Note that the switching criterion between ID mode and EH mode
depends on the receiver's condition such as the available energy
in the storage and the required processing or circuit power. In
this paper, we focus on the achievable rate and harvested energy
obtained by the transferred signals from both transmitters in the
IFC according to the different receiver modes. 
The mode switching policy based on the receiver's condition is
left as a future work. We assume that the mode decided by the
receiver is sent to both transmitters through the zero-delay and
error-free feedback link at the beginning of the frame.} We assume
that the transmitters have perfect knowledge of the CSI of their
associated links (i.e. the links between a transmitter and all
receivers) but do not share those CSI between them.
In addition, $M_t = M_r = M$ for simplicity, but it can be
extended to general antenna configurations. Assuming a frequency
flat fading channel, which is static over several frames, the
received signal ${\bf y}_i \in \mathbb{C}^{M \times 1}$ for $i=1,
2$ can be written as
\begin{eqnarray}\label{Sys_1}\nonumber
{\bf y}_1 = {\bf H}_{11} {\bf x}_1 +{\bf H}_{12} {\bf x}_2 + {\bf
n}_1,
\\
{\bf y}_2 = {\bf H}_{21} {\bf x}_1 +{\bf H}_{22} {\bf x}_2 + {\bf
n}_2,
\end{eqnarray}
where ${\bf n}_i \in \mathbb{C}^{M\times 1} $ is a complex white
Gaussian noise vector with a covariance matrix ${\bf I}_{M}$ and
${\bf H}_{ij} \in \mathbb{C}^{M\times M}$ is the normalized
frequency-flat fading channel from the $j$th transmitter to the
$i$th receiver such as $\sum_{l,k= 1}^{M} |h_{ij}^{(l,k)}|^2 =
\alpha_{ij}M$ \cite{Loyka}. Here, $h_{ij}^{(l,k)}$ is the $(l,
k)$th element of ${\bf H}_{ij}$ and $\alpha_{ij}\in [0,1]$. We
assume that $ {\bf H}_{ij}$ has a full rank. 
The vector ${\bf x}_j \in \mathbb{C}^{M \times 1}$ is the transmit
signal, in which the independent messages can be conveyed, at the
$j$th transmitter with a transmit power constraint for $j= 1$ and
$2$ as
\begin{eqnarray}\label{Sys_2}
E[\|{\bf x}_j\|^2] \leq P  {\text{ for }} j=1 \text{ and } 2.
\end{eqnarray}

When the receiver operates in ID mode, the achievable rate at
$i$th receiver, $R_i$, is given by \cite{Scutari}
\begin{eqnarray}\label{Sys_3_revised}
R_i = \log \det ({\bf I}_{M} +{\bf H}_{ii}^H{\bf R}_{-i}^{-1}{\bf
H}_{ii}{\bf Q}_i ),
\end{eqnarray}
where ${\bf R}_{-i}$ indicates the covariance matrix of noise and
interference at the $i$th receiver, i.e.,
\begin{eqnarray}\nonumber
{\bf R}_{-1} = {\bf I}_{M} + {\bf H}_{12}{\bf Q}_2 {\bf
H}_{12}^H,\\\nonumber {\bf R}_{-2} = {\bf I}_{M} + {\bf
H}_{21}{\bf Q}_1 {\bf H}_{21}^H.
\end{eqnarray}
Here, ${\bf Q}_j = E [{\bf x}_j {\bf x}_j^H ]$ denotes the
covariance matrix of the transmit signal at the $j$th transmitter
and, from (\ref{Sys_2}), $tr({\bf Q}_j) \leq P $.

For EH mode, it can be assumed that the total harvested power
$E_i$ at the $i$th receiver (more exactly, harvested energy
normalized by the baseband symbol period) is given by
\begin{eqnarray}\label{Sys_3}\nonumber
E_i &=& \zeta_i E[\|{\bf y}_i \|^2 ]\\
&=& \zeta_i tr\left( \sum_{j=1}^2{\bf H}_{ij}{\bf Q}_j{\bf
H}_{ij}^H +{\bf I}_{M} \right),
\end{eqnarray}
where $\zeta_i$ denotes the efficiency constant for converting the
harvested energy to electrical energy to be stored \cite{Vullers,
Zhang1}. For simplicity, it is assumed that $\zeta_i=1$ and the
noise power is negligible compared to the transferred energy from
each transmitters.\footnote{In this paper, we assume the system
operates in the high signal-to-noise ratio (SNR) regime, which is
also consistent with the practical wireless energy transfer
requires a high-power transmission, but we also discuss the low
SNR regime in Section \ref{sec:discussion} as well.} That is,
\begin{eqnarray}\label{Sys_4}E_i &\approx& tr\left( \sum_{j=1}^2{\bf
H}_{ij}{\bf Q}_j{\bf H}_{ij}^H \right)\nonumber\\
&=& tr\left( {\bf H}_{i1}{\bf Q}_1{\bf H}_{i1}^H \right)+tr\left(
{\bf H}_{i2}{\bf Q}_2{\bf H}_{i2}^H \right)\nonumber \\&=&
E_{i1}+E_{i2},
\end{eqnarray}
where $E_{ij}=tr\left( {\bf H}_{ij}{\bf Q}_j{\bf H}_{ij}^H
\right)$ denoting the energy transferred from the $j$th
transmitter to the $i$th receiver.

Interestingly, when the receiver decodes the information data from
the associated transmitter under the assumption that the signal
from the other transmitter is not decodable \cite{Shen}, the
signal from the other transmitter becomes an interference to be
defeated. In contrast, when the receiver harvests the energy, it
becomes a useful energy-transferring source. Fig.
\ref{JWIET_two_userIC_block} illustrates an example of the
receiving mode ($EH_1$, $ID_2$), where the interference1 (dashed
red line) should be reduced for ID, while the interference2
(dashed green line) be maximized for EH. In what follows, for four
possible receiving modes, we investigate the achievable
rate-harvested energy tradeoff. In addition, the corresponding
transmission strategy (more specifically, transmit signal design)
is presented.

\section{Two receivers on a single mode}\label{sec:single}

\subsection{Two IDs: maximum achievable sum rate}

For the scenario ($ID_1$, $ID_2$), it is desirable to obtain the
maximum achievable sum rate. That is, the problem can be
formulated as follows:

\begin{eqnarray}\label{twoID_1}
(P1) {\text{ maximize }}& \sum_{i=1}^{2} R_{i}\\\label{power_2_2}
{\text{subject to }}&tr({\bf Q}_j) \leq P,~ {\bf Q}_j \succeq{\bf
0}\quad {\text{for }} j = 1, 2,
\end{eqnarray}

The solution of (P1) has been extensively considered in many
previous communication researches \cite{Scutari,WeiYu,AGoldsmith},
where the iterative water-filling algorithms have been developed
to maximize the achievable rate in a distributed manner with no
CSI sharing between the transmitters. This is briefly summarized
in Algorithm 1:

\vspace*{10pt}Algorithm 1. {\it{\underline{Iterative
Water-filling:}}}
\begin{enumerate}
\item Initialize $n=0$ and ${\bf Q}_j^{(0)} \in \mathcal{Q}_P$ for
$j=1,2$, where
\begin{eqnarray}\label{twoID_2}
\mathcal{Q}_P \triangleq \{{\bf Q} \in \mathbb{C}^{M \times M}:
{\bf Q} \succeq {\bf 0}, tr({\bf Q}) = P\}.
\end{eqnarray}
\item For $n=0:N_{max}$, where $N_{max}$ is the maximum number of iterations\footnote{Generally, $N_{max} =20$ is sufficient for
the solutions to converge.}\\
~~~~~Update ${\bf Q}_j^{(n+1)}$ for $j=1,2$ as follows:
\begin{eqnarray}\label{twoID3}
\!{\bf Q}_j^{(n\!+\!1)}\!=\! \left\{\! \begin{array}{c} WF({\bf
H}_{jj}, {\bf R}_{-j}^{(n)},P),~\text{if ${\bf R}_{-j}^{(n)}$ is
updated,\!}\!
\\{\bf Q}_j^{(n)}, \quad \text{otherwise,}
\end{array} \!\right.\!\!
\end{eqnarray}
where ${\bf R}_{-j}^{(n)}$ indicates the covariance matrix of
noise and interference in the $j$th receiver at the $n$th
iteration, i.e.,
\begin{eqnarray}\nonumber
{\bf R}_{-1}^{(n)} = {\bf I}_{M} + {\bf H}_{12}{\bf Q}_2^{(n)}
{\bf H}_{12}^H,\\\nonumber {\bf R}_{-2}^{(n)} = {\bf I}_{M} + {\bf
H}_{21}{\bf Q}_1^{(n)} {\bf H}_{21}^H.
\end{eqnarray}
Note that ${\bf R}_{-j}^{(n)}$ is measured at each receiver
similarly to the way it has been done in \cite{Scutari} and,
furthermore, ${\bf Q}_j^{(n)}$ is computed at the receiver and
reported to the transmitter through the zero-delay and error-free
feedback link.
\item Finally, ${\bf Q}_j = {\bf Q}_j^{N_{max}+1}$ for $j=1,2$.
\end{enumerate}
 \vspace*{10pt}
Here, $WF()$ denotes the water-filling operator given as
\cite{Scutari}:
\begin{eqnarray}\label{twoID4}
WF({\bf H}_{ii}, {\bf R}, P) ={\bf U}_i (\mu_i{\bf I}_M - {\bf
D}_i^{-1})^+ {\bf U}_i^H,
\end{eqnarray}
where ${\bf U}_i$ and ${\bf D}_i$ are obtained from the eigenvalue
decomposition of ${\bf H}_{ii}^H{\bf R}^{-1}{\bf H}_{ii}$. That
is, $ {\bf H}_{ii}^H{\bf R}^{-1}{\bf H}_{ii} = {\bf U}_i {\bf D}_i
{\bf U}_i^H$, and $\mu_i$ denotes the water level that satisfies
the transmit power constraint as $tr\{(\mu_i{\bf I}_M - {\bf
D}_i^{-1})^+\} = P$.

In the scenario ($ID_1$, $ID_2$), because both receivers decode
the information, the harvested energy becomes zero.

\subsection{Two EHs: maximum harvested energy}

For the scenario ($EH_1$, $EH_2$), both receivers want to achieve
the maximum harvested energy. That is, the problem can be
formulated as:
\begin{eqnarray}\label{twoEH_1}
(P2) {\text{ maximize }}& \sum_{i=1}^{2} E_{i}\\\label{power_2_2}
{\text{subject to }}&tr({\bf Q}_j) \leq P,~ {\bf Q}_j \succeq{\bf
0}\quad {\text{for }} j = 1, 2,
\end{eqnarray}
The following proposition gives the optimal solution for the
problem (P2).

\begin{prop}\label{prop1} The optimal ${\bf Q}_j$ for (P2) has a rank equal to one and is given as ${\bf Q}_j = P [\bar{\bf V}_{j}]_1[\bar{\bf V}_{j}]_1^H$,
where $\bar{\bf V}_{j}$ is a $M \times M$ unitary matrix obtained
from the SVD of $\bar{\bf H}_{j} \triangleq
\left[\begin{array}{c}{\bf H}_{1j}\\{\bf H}_{2j}
\end{array} \right] $. That is, $\bar{\bf H}_{j} = \bar{\bf U}_{j} \bar{\bf \Sigma}_{j} \bar{\bf
V}_{j}^H$, where $\bar{\bf \Sigma}_{j} = diag\{\bar\sigma_{j, 1
},..., \bar\sigma_{j, M } \}$ with $\bar\sigma_{j, 1 }\geq...\geq
\bar\sigma_{j, M }$.
\end{prop}
\begin{proof}
From (\ref{Sys_4}),
\begin{eqnarray}\label{twoEH_2}
\sum_{i=1}^{2} E_{i} &=&\sum_{i=1}^{2} tr\left( \sum_{j=1}^2{\bf
H}_{ij}{\bf Q}_j{\bf H}_{ij}^H \right)\nonumber\\
&=&\sum_{j=1}^{2} tr\left( \sum_{i=1}^2{\bf H}_{ij}{\bf Q}_j{\bf
H}_{ij}^H \right)\nonumber\\
&=&\sum_{j=1}^{2} tr\left(  \bar{\bf H}_{j}{\bf Q}_j \bar{\bf
H}_{j}^H \right)
\end{eqnarray}
Note that the covariance matrix ${\bf Q}_j$ can be written as
${\bf Q}_j = {\bf V}_j {\bf D}_j^2 {\bf V}_j^H$ where ${\bf V}_j$
is a $M \times M$ unitary matrix and ${\bf D}_j^2 = diag\{
d_{j,1}^2,...,d_{j,M}^2 \}$ with $\sum_{m=1}^{M}d_{j,m}^2 \leq P
$. Because $tr({\bf A}{\bf B}) = tr({\bf B}{\bf A})$ for ${\bf
A}\in \mathbb{C}^{m\times n}$ and ${\bf B}\in \mathbb{C}^{n\times
m}$, (\ref{twoEH_2}) can be rewritten as
\begin{eqnarray}\label{twoEH3}
\sum_{i=1}^{2} E_{i} =\sum_{j=1}^{2} tr\left({\bf D}_j^2  {\bf
V}_j^H   \bar{\bf H}_{j}^H\bar{\bf H}_{j}{\bf V}_j \right) =
\sum_{j=1}^{2} \sum_{m=1}^{M} d_{j,m}^2 \|\bar{\bf H}_{j}[{\bf
V}_{j}]_m\|^2.
\end{eqnarray}
Because $\sum_{m=1}^{M}d_{j,m}^2 \leq P $,
\begin{eqnarray}\label{twoEH3_1}
\sum_{m=1}^{M} d_{j,m}^2 \|\bar{\bf H}_{j}[{\bf V}_{j}]_m\|^2 \leq
P \underset{m=1,...M}{\max} \|\bar{\bf H}_{j}[{\bf V}_{j}]_m\|^2.
\end{eqnarray}
Here, the equality holds when $d_{j,m'}^2 = P$ for $m' = \arg
\underset{m=1,...M}{\max} \|\bar{\bf H}_{j}[{\bf V}_{j}]_m\|^2$
and $d_{j,m}^2 = 0$ for $m\neq m'$, which implies that ${\bf Q}_j$
has a rank equal to one and accordingly, it is given as ${\bf
Q}_j= P [{\bf V}_{j}]_{m'}[{\bf V}_{j}]_{m'}^H$. Note that
\begin{eqnarray}\label{twoEH3_2}
\|\bar{\bf H}_{j}[{\bf V}_{j}]_{m'}\|^2 \leq \bar\sigma_{j, 1 }^2,
\end{eqnarray}
where the equality holds when $[{\bf V}_{j}]_{m'} = [\bar{\bf
V}_{j}]_1$. Therefore, from (\ref{twoEH3_1}) and (\ref{twoEH3_2}),
(\ref{twoEH3}) is bounded as
\begin{eqnarray}\label{twoEH4}
\!\sum_{i=1}^{2} E_{i} \!=\! \sum_{j=1}^{2} \sum_{m=1}^{M}
d_{j,m}^2 \|\bar{\bf H}_{j}[\bar{\bf V}_{j}]_m\|^2\!\leq P(
\bar\sigma_{1, 1 }^2 +\bar\sigma_{2, 1 }^2 ),\nonumber
\end{eqnarray}
and the equality holds when ${\bf Q}_j = P [\bar{\bf
V}_{j}]_1[\bar{\bf V}_{j}]_1^H$.
\end{proof}
Note that each transmitter can design the transmit covariance
matrix ${\bf Q}_j$ such that the transferred energy from each
transmitter is maximized without considering other transmitter's
channel information and transmission strategy. That is, thanks to
the energy conservation law, each transmitter transfers the energy
through its links independently.

From Proposition \ref{prop1}, the transmit signal on each
transmitter can be designed as ${\bf x}_j= \sqrt{P}[\bar{\bf
V}_{j}]_1{\bf s}_j$, where ${\bf s}_j$ is any random signal with
zero mean and unit variance. Because both receivers harvest the
energy and are not able to decode the information, the achievable
rate becomes zero.

\section{One ID receiver and One EH receiver}\label{sec:oneIDoneEH}

In this section, without loss of generality, we will consider
($EH_1$, $ID_2$) - the first receiver harvests the energy and the
second decodes information. The transmission strategy described
below can also be applied to ($ID_1$, $EH_2$) without difficulty.
Note that energy harvesting and information transfer occur
simultaneously in the IFC, and accordingly, the achievable
rate-energy region is not trivial compared to the scenarios
($EH_1$, $EH_2$) and ($ID_1$, $ID_2$).

\subsection{A necessary condition for the optimal transmission strategy}\label{ssec:optimal_tx_strat}

Because information decoding is done only at the second receiver,
by letting $R= R_2$ and $E = E_1 =E_{11}+E_{12}$, we can define
the achievable rate-energy region as:
\begin{eqnarray}\label{oneIDoneEH_1}
\!\!&\!C_{R\!-\!E} (P) \!\triangleq  \!\Biggl\{ \!(R, E) : R \leq
\log \det({\bf I}_{M} + {\bf H}_{22}^H{\bf R}_{-2}^{-1}{\bf
H}_{22}{\bf Q}_2 ),\!&\!\nonumber\\
\!\!&\!E \!\leq \!\sum_{\!j\!=\!1}^{\!2} tr ({\bf H}_{1j} {\bf
Q}_j {\bf H}_{1j}^H), tr({\bf Q}_j)\!\leq \!P, {\bf Q}_j\!\succeq
\!{\bf 0}, j\!=\!1,\!2\! \Biggr\}\!.\!&\!
\end{eqnarray}
Here, because EH and ID operations in the IFC interact with each
other, the boundary of the rate-energy region is not easily
characterized and is so far unknown.
The following lemma gives a useful insight into the
derivation of the optimal boundary.

\begin{lem}\label{lem1} For ${\bf H}_{11}$ and ${\bf H}_{21}$,
there always exists an invertible matrix ${\bf T}\in
\mathbb{C}^{M\times M}$ such that
\begin{eqnarray}\label{oneIDoneEH2}
{\bf U}_{G}^H{\bf H}_{11}{\bf T} = {\bf \Sigma}_G\nonumber \\
{\bf V}_{G}^H{\bf H}_{21}{\bf T} = {\bf I}_{M},
\end{eqnarray}
where ${\bf U}_{G}$ and ${\bf V}_{G}$ are unitary and ${\bf
\Sigma}_G$ is a diagonal matrix with $ \sigma_{G,1} \geq
\sigma_{G,2} \geq ,...,\geq \sigma_{G,M}\geq 0$.
\end{lem}
\begin{proof}
Because $ {\bf H}_{21}$ has a full rank, by utilizing the
generalized singular value decomposition \cite{Sadek, Paige}, we
can obtain an invertible matrix ${\bf T}'$ such that
\begin{eqnarray}\label{oneIDoneEH2_1}
{\bf U}_{G}^H{\bf H}_{11}{\bf T}' = {\bf \Sigma}_A\nonumber \\
{\bf V}_{G}^H{\bf H}_{21}{\bf T}' = {\bf \Sigma}_{B},\nonumber
\end{eqnarray}
where ${\bf U}_{G}$ and ${\bf V}_{G}$ are unitary and ${\bf
\Sigma}_A$ and ${\bf \Sigma}_B$ are diagonal matrices with $1 \geq
\sigma_{A,1} \geq \sigma_{A,2} \geq ,...,\geq \sigma_{A,M}\geq 0$
and with $0 < \sigma_{B,1} \leq \sigma_{B,2} \leq ,...,\leq
\sigma_{B,M}\leq 1$, respectively. Here, $\sigma_{A,i}^2
+\sigma_{B,i}^2 = 1$. Therefore, by setting ${\bf T} ={\bf T}'{\bf
\Sigma}_{B}^{-1}$, we can obtain (\ref{oneIDoneEH2}) with ${\bf
\Sigma}_G = {\bf \Sigma}_A{\bf \Sigma}_{B}^{-1}$.
\end{proof}
Without loss of generality, we set
\begin{eqnarray}\label{oneIDoneEH3_rev1}
{\bf Q}_{1} = {\bf T}{\bf X}{\bf X}^H {\bf T}^H,
\end{eqnarray}
where ${\bf X}\in \mathbb{C}^{M \times m}$ has the SVD as
\begin{eqnarray}\label{oneIDoneEH3}\nonumber
{\bf X} = {\bf U}_{x}{\bf \Sigma}_{x}{\bf V}_{x}^H
\end{eqnarray}
with ${\bf \Sigma}_{x} = diag\{\sigma_{x,1},...,\sigma_{x,m}\}$
and $\sigma_{x,1}\geq,...,\geq\sigma_{x,m}$. Here,
\begin{eqnarray}\label{oneIDoneEH3_1}
\sum_{i=1}^m \sigma_{x,i}^2 =P',
\end{eqnarray}
where $P'$ is a normalization constant such that $tr({\bf T}{\bf
X}{\bf X}^H {\bf T}^H) \leq P$ is satisfied. We then have the
following proposition.
\begin{prop}\label{prop2} In the high SNR regime, the optimal ${\bf Q}_1$ at the boundary
of the achievable rate-energy region has a rank one at most. That
is, $rank ({\bf Q}_1) \leq 1$.
\end{prop}
\begin{proof}
First, let us consider the boundary point ($\bar R$, $\bar E$) of
the achievable rate-energy, in which $\bar E \leq tr({\bf
H}_{12}{\bf Q}_2{\bf H}_{12}^H)$. Then, because the first
transmitter do not need to transmit any signals causing the
interference to the ID receiver (the second receiver), ${\bf Q}_1
= {\bf 0}$ is optimal. That is, $rank({\bf Q}_1) = 0$.

For $\bar E > tr({\bf H}_{12}{\bf Q}_2{\bf H}_{12}^H)$, let there
be ${\bf Q}_1$ with $m = rank({\bf Q}_1)>1$ which corresponds to
the boundary point ($\bar R$, $\bar E$) of the achievable
rate-energy. Then, given the harvested energy $\bar E$ (the
boundary point) and ${\bf Q}_2$, the covariance matrix ${\bf Q}_1$
exhibits
\begin{eqnarray}\label{oneIDoneEH4}
\bar R = \log \det({\bf I}_{M} + {\bf H}_{22}^H{\bf
R}_{-2}^{-1}{\bf H}_{22}{\bf Q}_2 )
\end{eqnarray}
with
\begin{eqnarray}\label{oneIDoneEH4_1}
tr({\bf H}_{11}{\bf Q}_1{\bf H}_{11}^H) = \bar E_{11},
\end{eqnarray}
where $\bar E_{11} \triangleq \bar E -tr({\bf H}_{12}{\bf Q}_2{\bf
H}_{12}^H)$. Because of Sylvester's determinant theorem
\cite{DHarville} ($\det({\bf I} +{\bf AB})= \det({\bf I} +{\bf
BA})$ ), by substituting ${\bf R}_{-2}={\bf I}_{M} + {\bf
H}_{21}{\bf Q}_1{\bf H}_{21}^H$ into (\ref{oneIDoneEH4}), we can
rewrite (\ref{oneIDoneEH4}) as
\begin{eqnarray}\label{oneIDoneEH4_2}
\bar R &=& \log \det({\bf I}_{M} + {\bf H}_{22}{\bf Q}_2{\bf
H}_{22}^H({\bf I}_{M} + {\bf H}_{21}{\bf Q}_1{\bf H}_{21}^H)^{-1}
)\nonumber \\&=& \log \det(({\bf I}_{M} + {\bf H}_{21}{\bf
Q}_1{\bf H}_{21}^H)({\bf I}_{M} + {\bf H}_{21}{\bf Q}_1{\bf
H}_{21}^H)^{-1} + {\bf H}_{22}{\bf Q}_2{\bf H}_{22}^H({\bf I}_{M}
+ {\bf H}_{21}{\bf Q}_1{\bf H}_{21}^H)^{-1} )\nonumber \\&=& \log
\det(({\bf I}_{M} + {\bf H}_{21}{\bf Q}_1{\bf H}_{21}^H)^{-1} +
({\bf H}_{21}{\bf Q}_1{\bf H}_{21}^H +{\bf H}_{22}{\bf Q}_2{\bf
H}_{22}^H)({\bf I}_{M} + {\bf H}_{21}{\bf Q}_1{\bf H}_{21}^H)^{-1}
)\nonumber \\&=& \log \det({\bf I}_{M} + {\bf H}_{21}{\bf Q}_1{\bf
H}_{21}^H + {\bf H}_{22}{\bf Q}_2{\bf H}_{22}^H ) - \log \det({\bf
I}_{M} + {\bf H}_{21}{\bf Q}_1{\bf H}_{21}^H).
\end{eqnarray}
Let us define $m_2 = rank({\bf Q}_2)$ and consider $m_2 \geq m$
without loss of generality. From Lemma 1 and
(\ref{oneIDoneEH3_rev1}), (\ref{oneIDoneEH4_2}) and
(\ref{oneIDoneEH4_1}) can be respectively rewritten as
\begin{eqnarray}\label{oneIDoneEH5_rev1}
\bar R &=&  \log \det({\bf I}_{M} + {\bf V}_{G}{\bf X}{\bf
X}^H{\bf V}_{G}^H + {\bf H}_{22}{\bf Q}_2{\bf H}_{22}^H ) - \log
\det({\bf I}_{M} + {\bf V}_{G}{\bf X}{\bf X}^H{\bf V}_{G}^H),\nonumber\\
&=&  \log \det({\bf I}_{M} +{\bf X}{\bf X}^H +  {\bf V}_{G}^H{\bf
H}_{22}{\bf Q}_2{\bf H}_{22}^H {\bf V}_{G} ) - \log \det({\bf
I}_{M} + {\bf X}{\bf X}^H),
\end{eqnarray}
and
\begin{eqnarray}\label{oneIDoneEH5_1}\nonumber
tr({\bf H}_{11}{\bf Q}_1{\bf H}_{11}^H) &=& tr({\bf H}_{11}{\bf
T}{\bf X}{\bf X}^H {\bf T}^H{\bf H}_{11}^H) = tr({\bf U}_{G}{\bf
\Sigma}_G{\bf X}{\bf X}^H {\bf \Sigma}^H{\bf U}_{G}^H)\\&=&
tr({\bf \Sigma}_G{\bf X}{\bf X}^H {\bf \Sigma}_G) =
\sum_{j=1}^{m}\sigma_{x,j}^2(\sum_{i=1}^M \sigma_{G,i}^2
|u_{x}^{(i,j)}|^2) =\bar E_{11},
\end{eqnarray}
where $u_{x}^{(i,j)}$ is the $(i,j)$th element of ${\bf U}_{x}$.
From the interlacing theorem (Theorem 3.1 in \cite{RHorn}),
(\ref{oneIDoneEH5_rev1}) can be further rewritten as
\begin{eqnarray}\label{oneIDoneEH5}
\bar R &=& \log\left({\prod_{i=1}^{m}(1+\sigma_{x,i}^2
+\kappa_{i}^2 )}\prod_{j=m+1}^{m_2}(1+\kappa_{j}^2)\right) - \log
{
\prod_{i=1}^{m}(1+\sigma_{x,i}^2)},\nonumber\\
&\approx& \log\left({\prod_{i=1}^{m}(\sigma_{x,i}^2 +\kappa_{i}^2
)}\prod_{j=m+1}^{m_2}\kappa_{j}^2\right) - \log {
\prod_{i=1}^{m}(1+ \sigma_{x,i}^2)},
\end{eqnarray}
where $\kappa_{j}^2$ is the interlaced value due to ${\bf
H}_{22}{\bf Q}_2{\bf H}_{22}^H$. That is, $\sigma_{y,m_2}^2\leq
\kappa_{j}^2 \leq \sigma_{y,1}^2$, $j=1,...,m_2$, where
$\sigma_{y,1}^2$ and $\sigma_{y,m_2}^2$ are the largest and the
smallest eigenvalues of ${\bf H}_{22}{\bf Q}_2{\bf H}_{22}^H$.
Note that the last approximation in (\ref{oneIDoneEH5}) is from
the high SNR regime (i.e., large power $P$ such that $\log (1+P)
\approx \log(P)$), where $\sigma_{y,i}^2$ for all $i=1,...,m_2$
are linearly proportional to $P$ resulting in
\begin{eqnarray}\label{oneIDoneEH5_2}
\kappa_{j}^2 \varpropto P \text{ for } j=1,..,m_2.
\end{eqnarray}
Because $ \sigma_{G,1} \geq ,...,\geq \sigma_{G,M}\geq 0$ and
\begin{eqnarray}\label{oneIDoneEH6}\nonumber
0 \leq |u_{x}^{(i,j)}|^2 \leq 1,\quad \sum_{i=1}^M
|u_{x}^{(i,j)}|^2 =1,
\end{eqnarray}
if there exists $m>1$ such that (\ref{oneIDoneEH5_1}) is satisfied
with (\ref{oneIDoneEH3_1}), we can find ${\bf Q}'_1$ with
$rank({\bf Q}'_1)=1$ satisfying (\ref{oneIDoneEH5_1}). In
addition, (\ref{oneIDoneEH5}) can be rewritten as:
\begin{eqnarray}
\bar R &\approx& \log\left(\frac{\prod_{i=1}^{m}(\sigma_{x,i}^2
+\kappa_{i}^2 )}{ \prod_{i=1}^{m}(1+
\sigma_{x,i}^2)}\prod_{j=m+1}^{m_2}\kappa_{j}^2\right)
\nonumber\\\label{oneIDoneEH7_1}
&=&\log\left(\prod_{i=1}^{m}\left( \frac{\sigma_{x,i}^2}{1+
\sigma_{x,i}^2} + \frac{\kappa_{i}^2}{1+
\sigma_{x,i}^2}\right)\prod_{j=m+1}^{m_2}\kappa_{j}^2\right)
\\\label{oneIDoneEH7_2}
&\approx&\log\left(\prod_{i=1}^{m}\left( \frac{\kappa_{i}^2}{1+
\sigma_{x,i}^2}\right)\prod_{j=m+1}^{m_2}\kappa_{j}^2\right)
\\\label{oneIDoneEH7_3}
&=&  \log\left(\frac{\prod_{i=1}^{m_2}\kappa_{i}^2 }{
\prod_{i=1}^{m}(1+\sigma_{x,i}^2)}\right).
\end{eqnarray}
The approximation in (\ref{oneIDoneEH7_2}) is from
(\ref{oneIDoneEH5_2}) with a large $P$. That is because
$\sigma_{x,i}^2$ is negligible with respect to $\kappa_i^2$ when
$\bar{E}$ is finite. From (\ref{oneIDoneEH3_1}),
$\prod_{i=1}^{m}(1+\sigma_{x,i}^2)$ in the denominator of
(\ref{oneIDoneEH7_3}) has the minimum value when $m=1$. In other
words, if ${\bf Q}_1$ with $m>1$ exhibits ($\bar R$, $\bar E$),
then we can find ${\bf Q}'_1$ with $m=1$ such that ($\bar R'$,
$\bar E$) with $\bar R'>\bar R$ in the high SNR regime, which
contradicts that the point ($\bar R$, $\bar E$) is a boundary
point.
\end{proof}

\begin{remark}\label{remark1}
Note that when the required harvested energy $\bar E$ (more
precisely, $\bar E_{11}$) is large, both $\sigma_{x,i}^2$ and
$\sigma_{y,i}^2$ are linearly proportional to $P$ resulting in
\begin{eqnarray}\label{Remark_eqn1}
\sigma_{x,i}^2,\kappa_{j}^2 \varpropto P \text{ for } i=1,...,m,
j=1,..,m_2.
\end{eqnarray}
Then, (\ref{oneIDoneEH7_1}) becomes:
\begin{eqnarray}\label{Remark_eqn2}
\bar R &\approx& \log\left(\prod_{i=1}^{m}\left(1+
\frac{\kappa_{i}^2}{
\sigma_{x,i}^2}\right)\prod_{j=m+1}^{m_2}\kappa_{j}^2\right).
\end{eqnarray}
Therefore,
\begin{eqnarray}\label{oneIDoneEH8}
\bar R & \varpropto & \log P^{m_2 - m} = (m_2 - m)\log P,
\end{eqnarray}
which implies that in the high SNR regime with large harvesting
energy $E$, the achievable rate is linearly proportional to $(m_2
- m)$. Then, we can easily find that it is maximized when $m=1$.
Note that it can be interpreted as the {\it{degree of freedom
(DOF)}} in the IFC \cite{Cadambe1}, in which by reducing the rank
of the transmit signal at the first transmitter, the DOF at the
second transceiver can be increased.
\end{remark}
\begin{remark}\label{remark2}
Intuitively, from the power transfer point of view, ${\bf Q}_1$
should be as close to the dominant eigenvector of ${\bf
H}_{11}^H{\bf H}_{11}$ as possible, which implies that the rank
one is optimal for power transfer. From the information transfer
point of view, when SNR goes to infinity, the rate maximization is
equivalent to the DOF maximization. That is, a larger rank for
${\bf Q}_1$ means that more dimensions at the second receiver will
be interfered. Therefore, a rank one for ${\bf Q}_1$ is optimal
for both information and power transfer.
\end{remark}

When each node has a single antenna ($M=1$), the scalar weight at
the $j$th transmitter can be written as $\sqrt{P_j} e^{j\theta_j}$
or simply, ${\bf Q}_j = P_j$. The achievable rate-energy region
can then be given as
\begin{eqnarray}\label{oneIDoneEH_1_rev1}
\!\!&\!C_{R\!-\!E} (P) \!\triangleq  \!\Biggl\{ \!(R, E) : R \leq
\log (1 + \frac{P_2|{h}_{22}|^2}{1+ P_1|{h}_{21}|^2} ),\!&\!\nonumber\\
\!\!&\!E \!\leq \! P_1|{h}_{11}|^2 + P_2|{h}_{12}|^2, P_j\!\leq
\!P, j\!=\!1,\!2\! \Biggr\}\!.\!&\!
\end{eqnarray}
From (\ref{oneIDoneEH_1_rev1}), we can easily find that $P_2= P$
at the boundary of the achievable rate-energy region. That is, the
second transmitter always transmits its signal with full power
$P$. Therefore, the optimal transmission strategy for $M=1$ boils
down to the power allocation problem of the first transmitter in
the IFC.

 From Proposition \ref{prop2}, when transferring the energy in
the IFC, the transmitter's optimal strategy is either a rank-one
beamforming or no transmission according to the energy transferred
from the other transmitter, which increases the harvested energy
at the corresponding EH receiver and simultaneously reduce the
interference at the other ID receiver. Even though the
identification of the optimal achievable R-E boundary is an open
problem, it can be found that the first transmitter will opt for a
rank-one beamforming scheme. Therefore, in what follows, we first
design two different rank-one beamforming schemes for the first
transmitter and identify the achievable rate-energy trade-off
curves for the two-user MIMO IFC where the rank-one beamforming
schemes are exploited.



\subsection{Rank-one Beamforming Design}\label{sec:rankoneBF}
\subsubsection{Maximum-energy beamforming (MEB)}\label{sssec:MEB}
Because the first receiver operates as an energy harvester, the
first transmitter may steer its signal to maximize the energy
transferred to the first receiver, resulting in a considerable
interference to the second receiver operating as an information
decoder.

From Proposition \ref{prop2}, the corresponding transmit
covariance matrix ${\bf Q}_1$ is then given by
\begin{eqnarray}\label{MEB}
{\bf Q}_1 = P_1 [{\bf V}_{11}]_1[{\bf V}_{11}]_1^H,
\end{eqnarray}
where ${\bf V}_{11}$ is a $M \times M$ unitary matrix obtained
from the SVD of ${\bf H}_{11}$ and $0\leq P_1 \leq P$. That is,
${\bf H}_{11} = {\bf U}_{11} {\bf \Sigma}_{11} {\bf V}_{11}^H$,
where ${\bf \Sigma}_{11} = diag\{\sigma_{11, 1 },..., \sigma_{11,
M } \}$ with $\sigma_{11, 1}\geq...\geq \sigma_{11, M }$. Here,
the energy harvested from the first transmitter is given by $P_1
\sigma_{11, 1 }^2$.

\subsubsection{Minimum-leakage beamforming (MLB)}\label{sssec:MLB} From an ID perspective
at the second receiver, the first transmitter should steer its
signal to minimize the interference power to the second receiver.
That is, from Proposition \ref{prop2}, the corresponding transmit
covariance matrix ${\bf Q}_1$ is then given by
\begin{eqnarray}\label{MLB}
{\bf Q}_1 = P_1 [{\bf V}_{21}]_M[{\bf V}_{21}]_M^H,
\end{eqnarray}
where ${\bf V}_{21}$ is a $M \times M$ unitary matrix obtained
from the SVD of ${\bf H}_{21}$ and $0\leq P_1 \leq P$. That is,
${\bf H}_{21} = {\bf U}_{21} {\bf \Sigma}_{21} {\bf V}_{21}^H$,
where ${\bf \Sigma}_{21} = diag\{\sigma_{21, 1 },..., \sigma_{21,
M } \}$ with $\sigma_{21, 1 }\geq...\geq \sigma_{21, M }$. Then,
the energy harvested from the first transmitter is given by $P_1
\|{\bf H}_{11} [{\bf V}_{21}]_M\|^2$.

%

\subsection{Achievable R-E region}\label{ssec:REregion}
Given ${\bf Q}_1$ as in either (\ref{MEB}) or (\ref{MLB}), the
achievable rate-energy region is then given as:
\begin{eqnarray}\label{oneIDoneEHmax1}
\!&\!C_{R-E} (P) =  \Biggl\{ (R, E) : R = R_2, E =E_{11} + E_{12}, \quad \!&\!\nonumber\\
\!&\!R_2 \!\leq \!\log \det({\bf I}_{M} + {\bf H}_{22}^H{\bf
R}_{-2}^{-1}{\bf H}_{22}{\bf Q}_2 ), E_{12} \!\leq \! tr ({\bf
H}_{12} {\bf Q}_2 {\bf H}_{12}^H), \!&\!\nonumber\\\!\!&\!\!
tr({\bf Q}_2)\leq P, {\bf Q}_2\succeq {\bf 0}
\Biggr\},\!\!\!&\!\!\!
\end{eqnarray}
where
\begin{eqnarray}\label{oneIDoneEHmax1_1}
E_{11} = \left\{\begin{array}{c}P_1 \sigma_{11, 1 }^2\quad  {\text{for MEB}}\\
P_1 \|{\bf H}_{11} [{\bf V}_{21}]_M\|^2\quad {\text{for
MLB}}\end{array} \right. ,
\end{eqnarray}
and
\begin{eqnarray}\label{oneIDoneEHmax1_2}
{\bf R}_{-2} = \left\{\begin{array}{c}{\bf I}_{M} + P_1 {\bf
H}_{21}[{\bf V}_{11}]_1[{\bf V}_{11}]_1^H {\bf H}_{21}^H  \quad{\text{for MEB}}\\
{\bf I}_{M} + P_1\sigma_{21, M }^2 [{\bf U}_{21}]_M [{\bf
U}_{21}]_M^H\quad {\text{for MLB}}\end{array} \right. .
\end{eqnarray}
Note that because $\sigma_{11, 1 }^2 \geq \|{\bf H}_{11} [{\bf
V}_{21}]_M\|^2$, the energy harvested by the first receiver from
the first transmitter with MEB is generally larger than that with
MLB.

Due to Sylvester's determinant theorem, $R_2$ can be derived as:
\begin{eqnarray}\label{oneIDoneEHmax2}
R_2 &=& \log \det ({\bf I}_{M} +{\bf H}_{22}^H{\bf
R}_{-2}^{-1}{\bf H}_{22}{\bf Q}_2 )\nonumber \\&=& \log \det ({\bf
I}_{M} + {\bf R}_{-2}^{-1/2}{\bf H}_{22}{\bf Q}_2{\bf
H}_{22}^H{\bf R}_{-2}^{-1/2} ).
\end{eqnarray}
Accordingly, by letting $\tilde{\bf H}_{22} = {\bf
R}_{-2}^{-1/2}{\bf H}_{22}$, we have the following optimization
problem for the rate-energy region of
(\ref{oneIDoneEHmax1})\footnote{The dual problem of maximizing
energy subject to rate constraint can be formulated, but the rate
maximization problem is preferred because it can be solved using
approaches similar to those in the rate maximization problems
under various constraints \cite{Scutari, SBoyd, Zhang1}.}
\begin{eqnarray}\label{oneIDoneEHmax3}
\!\!(\!P3\!)\! \underset{{\bf Q}_2}{\text{ maximize}}& \log \det
({\bf I}_{M} + \tilde{\bf H}_{22}{\bf Q}_2\tilde{\bf H}_{22}^H
)\\\label{oneIDoneEHmax3_1} \!\!{\text{subject to}}\!&\!\!tr({\bf
H}_{12}{\bf Q}_2{\bf H}_{12}^H) \geq \max(\bar E \!-\!
E_{11},0)\!\\\label{oneIDoneEHmax3_2} &tr({\bf Q}_2) \leq P,~{\bf
Q}_2 \succeq{\bf 0},
\end{eqnarray}
where $\bar E$ can take any value less than $E_{\max}$ denoting
the maximum energy transferred from both transmitters. Here, it
can be easily derived that $E_{\max}$ is given as
\begin{eqnarray}\label{oneIDoneEHmax4}
E_{\max} = \left\{\begin{array}{c} P (\sigma_{11, 1 }^2 +
\sigma_{12, 1 }^2)\quad {\text{for MEB}}\\P (\|{\bf H}_{11} [{\bf
V}_{21}]_M\|^2 + \sigma_{12, 1 }^2)\quad{\text{for MLB}}
\end{array}\right. ,
\end{eqnarray}
where $\sigma_{12, 1 }^2$ is the largest singular value of ${\bf
H}_{12}$ and it is achieved when the second transmitter also
steers its signal such that its beamforming energy is maximized on
the cross-link channel ${\bf H}_{12}$. That is,
\begin{eqnarray}\label{oneIDoneEHmax4_1}
{\bf Q}_2 = P [{\bf V}_{12}]_1[{\bf V}_{12}]_1^H.
\end{eqnarray}
Note that the corresponding transmit signal is given by $x_2(n)=
\sqrt{P}[{\bf V}_{12}]_1 s_2(n) $, where $s_2(n)$ is a random
signal with zero mean and unit variance. Therefore, when $s_2(n)$
is Gaussian randomly distributed with zero mean and unit variance,
which can be realized by using a Gaussian random code
\cite{TCover}, the achievable rate is given by $ R_2 = \log \det
({\bf I}_{M} + P\tilde{\bf H}_{22}[{\bf V}_{12}]_1 [{\bf
V}_{12}]_1^H\tilde{\bf H}_{22}^H ).$

Note that because $E_{11}$ in (\ref{oneIDoneEHmax3_1}) and
$\tilde{\bf H}_{22}$ in (\ref{oneIDoneEHmax3}) are dependent on
$P_1(\leq P)$, we identify the achievable R-E region iteratively
as:

\vspace*{10pt}Algorithm 2. {\it{\underline{Identification of the
achievable R-E region:}}}
\begin{enumerate}
\item Initialize $n=0$, $P_1^{(0)}=P$,
\begin{eqnarray}\label{ReviseAlgo2_1}
E_{11}^{(0)} = \left\{\begin{array}{c}P_1^{(0)} \sigma_{11, 1 }^2\quad  {\text{for MEB}}\\
P_1^{(0)} \|{\bf H}_{11} [{\bf V}_{21}]_M\|^2\quad {\text{for
MLB}}\end{array} \right. ,
\end{eqnarray}
and
\begin{eqnarray}\label{ReviseAlgo2_2}
{\bf R}_{-2}^{(0)} = \left\{\begin{array}{c}{\bf I}_{M} +
P_1^{(0)} {\bf
H}_{21}[{\bf V}_{11}]_1[{\bf V}_{11}]_1^H {\bf H}_{21}^H  \quad{\text{for MEB}}\\
{\bf I}_{M} + P_1^{(0)}\sigma_{21, M }^2 [{\bf U}_{21}]_M [{\bf
U}_{21}]_M^H\quad {\text{for MLB}}\end{array} \right. .
\end{eqnarray}
\item For $n=0:N_{max}$\\
~~~~~Solve the optimization problem (P3) for ${\bf Q}_2^{(n)}$ as
a function of
$E_{11}^{(n)}$ and ${\bf R}_{-2}^{(n)}$.\\
~~If $tr ({\bf H}_{12} {\bf Q}_2^{(n)} {\bf H}_{12}^H) +
E_{11}^{(n)}
>\bar E$
\begin{eqnarray}\label{eqnRevise_Algo_2_2} P_1^{(n+1)} =
max\left(\frac{\bar E - tr ({\bf H}_{12} {\bf Q}_2^{(n)} {\bf
H}_{12}^H)}{\kappa}, 0\right),
\end{eqnarray}
where $\kappa = \left\{\begin{array}{c} \sigma_{11, 1 }^2\quad  {\text{for MEB}}\\
\|{\bf H}_{11} [{\bf V}_{21}]_M\|^2\quad {\text{for
MLB}}\end{array} \right.$.\\
~~~~~Then, $P_1^{(n+1)} = min(P, P_1^{(n+1)})$ and update
$E_{11}^{(n+1)}$ and ${\bf R}_{-2}^{(n+1)}$ with $P_1^{(n+1)}$
similarly to (\ref{ReviseAlgo2_1}) and (\ref{ReviseAlgo2_2}).
\item Finally, the boundary point of the achievable R-E region is given as $(R, E) =(\log \det
({\bf I}_{M} + \tilde{\bf H}_{22}{\bf Q}_2^{(N_{max}+1)}\tilde{\bf
H}_{22}^H ),~ E_{11}^{(N_{max}+1)} + tr ({\bf H}_{12} {\bf
Q}_2^{(N_{max}+1)} {\bf H}_{12}^H)) $.
\end{enumerate}
 \vspace*{10pt}
In (\ref{eqnRevise_Algo_2_2}), if the total transferred energy
($tr ({\bf H}_{12} {\bf Q}_2^{(n)} {\bf H}_{12}^H) +
E_{11}^{(n)}$) is larger than the required harvested energy $\bar
E$, the first transmitter reduces the transmit power $P_1$ to
lower the interference to the ID receiver. In addition, if the
energy harvested by the first receiver from the second transmitter
($tr ({\bf H}_{12} {\bf Q}_2^{(n)} {\bf H}_{12}^H)$) is larger
than $\bar E$, the first transmitter does not transmit any signal.
That is, $rank({\bf Q}) =0$ as claimed in the proof of Proposition
\ref{prop2}.

To complete Algorithm 2, we now show how to solve the optimization
problem (P3) for ${\bf Q}_2^{(n)}$ in Step 2 of Algorithm 2. The
optimization problem (P3) with $E_{11}^{(n)}$ and ${\bf
R}_{-2}^{(n)}$ can be tackled with two different approaches
according to the value of $\bar E$, i.e., $0 \leq \bar E \leq
E_{11}$ and $E_{11} < \bar E \leq E_{\max}$. Note that we have
dropped the superscript of the iteration index $(n)$ for notation
simplicity. For $0 \leq \bar E \leq E_{11}$, (P3) becomes the
conventional rate maximization problem for single-user effective
MIMO channel (i.e., $\tilde{\bf H}_{22}$ ) whose solution is given
as
\begin{eqnarray}\label{oneIDoneEHmax5}
{\bf Q}_2 &=& WF(\tilde{\bf H}_{22},{\bf I}_M, P),
\end{eqnarray}
resulting in the maximum achievable rate for the given rank-one
strategy ${\bf Q}_1$. Here, the operator $WF()$ is defined in
(\ref{twoID4}).

For $E_{11} < \bar E \leq E_{\max}$, the optimization problem (P3)
can be solved by a ``water-filling-like'' approach similar to the
one appeared in the joint wireless information and energy
transmission optimization with a single transmitter \cite{Zhang1}.
That is, the Lagrangian function of ($P3$) can be written as
\begin{eqnarray}\label{oneIDoneEHmax8}\nonumber
\!&L({\bf Q}_2, \lambda, \mu) = \log \det ({\bf I}_{M} +
\tilde{\bf H}_{22}{\bf Q}_2\tilde{\bf H}_{22}^H
)\quad\quad\quad\quad\quad\quad\quad\quad\quad&
\\\nonumber&  +\lambda( tr({\bf H}_{12}{\bf Q}_2{\bf H}_{12}^H) -
(\bar E \!-\! E_1)  ) - \mu (tr({\bf Q}_2) - P),&
\end{eqnarray}
and the corresponding dual function is then given by \cite{SBoyd,
Zhang1}
\begin{eqnarray}\label{oneIDoneEHmax9}
g(\lambda, \mu) = \underset{{\bf Q}_2 \succeq{\bf 0}}{\max} L({\bf
Q}_2, \lambda, \mu).
\end{eqnarray}
Here the optimal solution $\mu'$, $\lambda'$, and ${\bf Q}_2$ can
be found through the iteration of the following steps \cite{SBoyd}
\begin{enumerate}
\item The maximization of $L({\bf Q}_2, \lambda, \mu)$ over ${\bf Q}_2$ for given $\lambda, \mu$.
\item The minimization of $g(\lambda, \mu)$ over $\lambda, \mu$ for given ${\bf Q}_2$.
\end{enumerate}
Note that, for given $\lambda, \mu$, the maximization of $L({\bf
Q}_2, \lambda, \mu)$ can be simplified as
\begin{eqnarray}\label{oneIDoneEHmax9_1}
\!&\underset{{\bf Q}_2 \succeq{\bf 0}}{\max} L({\bf Q}_2, \lambda,
\mu) = \log \det ({\bf I}_{M} + \tilde{\bf H}_{22}{\bf
Q}_2\tilde{\bf H}_{22}^H)  - tr({\bf A}{\bf Q}_2 ),&
\end{eqnarray}
where ${\bf A}= {\mu}{\bf I}_M -\lambda{\bf H}_{12}^H{\bf
H}_{12}$. Note that (\ref{oneIDoneEHmax9_1}) is the point-to-point
MIMO capacity optimization with a single weighted power constraint
and the solution is then given by \cite{SBoyd, Zhang1}
\begin{eqnarray}\label{oneIDoneEHmax6}
{\bf Q}_2 &=& {\bf A}^{-1/2}\tilde{\bf V}'_{22}\tilde{\boldsymbol
\Lambda}'\tilde{\bf V}_{22}'^H{\bf A}^{-1/2},
\end{eqnarray}
where $\tilde{\bf V}'_{22}$ is obtained from the SVD of the matrix
$\tilde{\bf H}_{22}{\bf A}^{-1/2}$, i.e., $\tilde{\bf H}_{22}{\bf
A}^{-1/2} =\tilde{\bf U}'_{22}\tilde{\boldsymbol \Sigma}'_{22}
\tilde{\bf V}_{22}'^H$. Here, $\tilde{\boldsymbol \Sigma}'_{22} =
diag\{ \tilde\sigma_{22,1}',...,\tilde\sigma_{22,M}' \}$ with $
\tilde\sigma_{22,1}' \geq...\geq \tilde\sigma_{22,M}' \geq 0$ and
$\tilde{\boldsymbol \Lambda}' = diag\{ \tilde p_1,...,\tilde p_M
\}$ with $\tilde p_i = (1-1/\tilde\sigma_{22,i}'^2)^+ $,
$i=1,...,M$. The parameters $\mu$ and $\lambda$ minimizing
$g(\lambda, \mu)$ in Step 2 can be solved by the subgradient-based
method \cite{Zhang1, XZhao}, where the the subgradient of
$g(\lambda, \mu)$ is given by $(tr({\bf H}_{12}{\bf Q}_2{\bf
H}_{12}^H) - (\bar E \!-\! E_1),  P - tr({\bf Q}_2)) $.

%
%

\section{Energy-regularized SLER-maximizing beamforming}
\label{sec:SLER_maximizing}
In Section \ref{sec:rankoneBF}, two rank-one beamforming
strategies are developed according to different aims - either
maximizing transferred energy to EH or minimizing interference
(or, leakage) to ID. Note that in \cite{Sadek, JPark}, the
maximization of the ratio of the desired signal power to leakage
of the desired signal on other users plus noise measured at the
transmitter, i.e., SLNR maximization, has been utilized in the
beamforming design in the multi-user MIMO system. Similarly, in
this section, to maximize transferred energy to EH and
simultaneously minimize the leakage to ID, we define a new
performance metric, signal-to-leakage-and-harvested energy ratio
(SLER) as
\begin{eqnarray}\label{GSVD1}
SLER =\frac{\|{\bf H}_{11}{\bf v} \|^2}{\|{\bf H}_{21}{\bf v} \|^2
+ max(\bar E -P_1\|{\bf H}_{11}\|^2 ,0)}.
\end{eqnarray}
Note that the noise power contributes to the denominator of SLNR
in the beamforming design \cite{Sadek, JPark} because the noise at
the receiver affects the detection performance degradation for
information transfer. That is, the noise power should be
considered in the computation of beamforming weights. In contrast,
the contribution of the {\it{minimum required harvested energy}}
is added in SLER of (\ref{GSVD1}), because the required harvested
energy minus the energy directly harvested from the first
transmitter is a main performance barrier of the EH receiver.
Therefore, in the energy beamforming, the required harvested
energy is considered in the computation of the beamforming
weights.\footnote{Strictly speaking, the SLER can be defined as
$SLER =\frac{\|{\bf H}_{11}{\bf v} \|^2}{\|{\bf H}_{21}{\bf v}
\|^2 + max(\bar E -\|{\bf H}_{11}{\bf v}\|^2 ,0)}$. However, for
computational simplicity, the lower bound on the required
harvested energy is added in the denominator of SLER from the fact
that $\|{\bf H}_{11}{\bf v}\|^2 \leq P_1\|{\bf H}_{11}\|^2$.}
Then, the SLER of (\ref{GSVD1}) can be rewritten as
\begin{eqnarray}\label{GSVD2}
SLER =\frac{{\bf v}^H{\bf H}_{11}^H{\bf H}_{11}{\bf v} }{{\bf
v}^H\left({\bf H}_{21}^H{\bf H}_{21} +{max({\bar E}/{P_1} -\|{\bf
H}_{11}\|^2 ,0)}{\bf I}_M \right){\bf v} }.
\end{eqnarray}
The beamforming vector ${\bf v}$ that maximizes SLER of
(\ref{GSVD2}) is then given by
\begin{eqnarray}\label{GSVD3}
{\bf v} = \frac{\sqrt{P_1}}{\|\bar{\bf v}\|} \bar{\bf v},
\end{eqnarray}
where $\bar{\bf v}$ is the generalized eigenvector associated with
the largest generalized eigenvalue of the matrix pair $({\bf
H}_{11}^H{\bf H}_{11}, {\bf H}_{21}^H{\bf H}_{21} +{max(\bar E
/{P_1} -\|{\bf H}_{11}\|^2 ,0)}{\bf I}_M) $. Here, $\bar{\bf v}$
can be efficiently computed by using a generalized singular value
decomposition (GSVD) algorithm \cite{JPark, JPark2}, which is
briefly summarized in Algorithm 3.

\vspace*{10pt}Algorithm 3. {\it{\underline{SLER maximizing
GSVD-based beamforming:}}}
\begin{enumerate}
\item Set ${\bf K}= \left[\begin{array}{c}{\bf H}_{11}\\ {\bf H}_{21}\\\sqrt{{max\left(\bar
E /{P_1} -\|{\bf H}_{11}\|^2 ,0\right)}}{\bf
I}_{M}\end{array}\right] \in \mathbb{C}^{3M\times M}$.
\item Compute QR decomposition (QRD) of ${\bf
K}\left(=\left[{\bf P}_{\alpha};{\bf P}_{\beta}  \right]{\bar{\bf
R}}\right)$, where $\left[{\bf P}_{\alpha};{\bf P}_{\beta}
\right]$ is unitary and $\bar{ \bf R}\in \mathbb{C}^{M\times M}$
is upper triangular. Here, ${\bf P}_{\alpha}\in
\mathbb{C}^{2M\times M}$.
\item Compute $\bar{ \bf V}_\alpha$ from the SVD of ${\bf
P}_\alpha$, i.e., $\bar{\bf U}_\alpha^{H} ({\bf P}_\alpha)_{1:M}
\bar{\bf V}_\alpha = {\bar{\bf \Sigma}}_\alpha$.
\item $\bar{\bf v} = \bar{\bf R}^{-1}[\bar{\bf V}_{\alpha}]_1$ and then, ${\bf v} = \frac{\sqrt{P_1}}{\|\bar{\bf v}\|}
\bar{\bf v}$.
\end{enumerate}
 \vspace*{10pt}
Here, because
\begin{eqnarray}\nonumber
\!\!{\bf K}\!= \!\left[\!\begin{array}{c}{\bf H}_{11}\\ {\bf
H}_{21}\\\sqrt{{max\left(\bar E /{P_1}-\|{\bf H}_{11}\|^2
,0\right)}}{\bf I}_{M}\end{array}\!\right] \!= \!\left[{\bf
P}_{\alpha};{\bf P}_{\beta} \right]{\bar{\bf R}}\!
\end{eqnarray}
as in \cite{JPark2}, for $P_1\|{\bf H}_{11}\|^2< \bar E$
\begin{eqnarray}\label{eqn_algo_step4}
\bar{\bf R}^{-1} = \frac{1}{\sqrt{{max\left(\bar E /{P_1}-\|{\bf
H}_{11}\|^2 ,0\right)}}} {\bf P}_{\beta},
\end{eqnarray}
which avoids a matrix inversion in Step 4 of Algorithm 3. Because
Algorithm 3 requires one QRD of an $3M\times M$ matrix (Step 2),
one SVD of an $M\times M$ matrix (Step3), and one ($M\times M$,
$M\times 1$) matrix-vector multiplication (Step 4 with
(\ref{eqn_algo_step4})), it has a slightly more computational
complexity compared to the MEB and the MLB in Section
\ref{sec:rankoneBF} that need one SVD of an $M\times M$ matrix. 

Once the beamforming vector is given as (\ref{GSVD3}), we can
obtain the R-E tradeoff curve for SLER maximization beamforming by
taking the approach described in Section \ref{ssec:REregion}.
Interestingly, from (\ref{GSVD2}), when the required harvested
energy at the EH receiver is large, the matrix in the denominator
of (\ref{GSVD1}) approaches an identity matrix multiplied by a
scalar. Accordingly, the SLER maximizing beamforming is equivalent
with the MEB in Section \ref{sssec:MEB}. That is, ${\bf v}$
becomes $\sqrt{P_1}[{\bf V}_{11}]_1$. In contrast, as the required
harvested energy becomes smaller, ${\bf v}$ is steered such that
less interference is leaked into the ID receiver to reduce the
denominator of (\ref{GSVD1}). That is, ${\bf v}$ approaches the
MLB weight vector in Section \ref{sssec:MLB}. Therefore, the
proposed SLER maximizing beamforming weighs up both metrics -
energy maximization to EH and leakage minimization to ID.


Note that the SLER value indicates how suitable a receiving mode,
($EH_1$, $ID_2$) or ($ID_1$, $EH_2$), is to the current channel.
This motivates us to propose a mode scheduling between ($EH_1$,
$ID_2$) and ($ID_1$, $EH_2$). That is, higher SLER implies that
the transmitter can transfer more energy to its associated EH
receiver incurring less interference to the ID receiver. Based on
this observation, our scheduling process can start with evaluating
for a given interference channel and $P$,
\begin{eqnarray}\label{GSVD4}
SLER^{(1)} =\underset{{\bf v}}{\max}\frac{\|{\bf H}_{11}{\bf v}
\|^2}{\|{\bf H}_{21}{\bf v} \|^2 + max(\bar E -P\|{\bf H}_{11}\|^2
,0)}
\end{eqnarray}
and
\begin{eqnarray}\label{GSVD5}
SLER^{(2)} =\underset{{\bf v}}{\max}\frac{\|{\bf H}_{22}{\bf v}
\|^2}{\|{\bf H}_{12}{\bf v} \|^2 + max(\bar E -P\|{\bf H}_{22}\|^2
,0)}.
\end{eqnarray}
If $SLER^{(1)} \geq SLER^{(2)}$, ($EH_1$, $ID_2$) is selected.
Otherwise, ($ID_1$, $EH_2$) is selected.

\section{Discussion}
\label{sec:discussion}
\subsection{The rank-one optimality in the low SNR regime for one ID receiver and one EH receiver}\label{ssec:LowSNRregime}
Even though we have assumed the high SNR regime throughout the
paper, in some applications such as wireless ad-hoc sensor
networks, low power transmissions are also considered. The
following proposition establishes the rank-one optimality in the
low SNR regime.

\begin{prop}\label{prop3} Considering ($EH_1$, $ID_2$) without loss of generality, in the low SNR regime, the optimal ${\bf Q}_1$ at the boundary
of the achievable rate-energy region has a rank one.
\end{prop}
\begin{proof}
Similarly to (\ref{oneIDoneEH4_2}), the achievable rate at the
$ID_2$ receiver is given by
\begin{eqnarray}\label{discuss1}
\bar R = \log \det({\bf I}_{M} + {\bf H}_{21}{\bf Q}_1{\bf
H}_{21}^H + {\bf H}_{22}{\bf Q}_2{\bf H}_{22}^H ) - \log \det({\bf
I}_{M} + {\bf H}_{21}{\bf Q}_1{\bf H}_{21}^H).
\end{eqnarray}
For a Hermitian matrix ${\bf A}$ with eigenvalues in $(-1, 1)$,
$\log \det ({\bf I}+{\bf A})$ can be extended as \cite{RHGohary}
\begin{eqnarray}\label{discuss2}
\log \det ({\bf I}+{\bf A}) = tr({\bf A}) - \frac{1}{2}tr({\bf
A}^2)+ \frac{1}{3}tr({\bf A}^3)+....
\end{eqnarray}
Because ${\bf H}_{ij}{\bf Q}_j{\bf H}_{ij}^H$ is Hermitian and
positive definite, and its maximum eigenvalue is upper-bounded as
$\lambda_{\max}({\bf H}_{ij}{\bf Q}_j{\bf
H}_{ij}^H)<\lambda_{\max}({\bf Q}_j)\lambda_{\max}({\bf
H}_{ij}{\bf H}_{ij}^H)$ \cite{RHGohary}, for sufficiently low
transmission power, their maximum eigenvalues lie in $(-1,1)$.
Accordingly, in the low SNR regime, $\bar R$ can be approximated
as
\begin{eqnarray}\label{discuss3}
\bar R &\approx& tr({\bf H}_{21}{\bf Q}_1{\bf H}_{21}^H +{\bf
H}_{22}{\bf Q}_2{\bf H}_{22}^H) -  tr({\bf H}_{21}{\bf Q}_1{\bf
H}_{21}^H)\nonumber \\ & = &  tr({\bf H}_{22}{\bf Q}_2{\bf
H}_{22}^H).
\end{eqnarray}
That is, the achievable rate is independent of the interference
from the first transmitter (noise-limited system). Then, ${\bf
Q}_1$ at the first transmitter can be designed to maximize the
harvested energy. Therefore, the optimal ${\bf Q}_1$ at the
boundary of the achievable rate-energy region is given by
\begin{eqnarray}\label{discuss4}
{\bf Q}_1 = \arg \underset{{\bf Q}}{\max}~ tr({\bf H}_{11}{\bf
Q}{\bf H}_{11}).
\end{eqnarray}
Note that $tr({\bf H}_{11}{\bf Q}{\bf H}_{11}) \leq P
\sigma_{11,1}^2 $, where the equality is satisfied when ${\bf Q} =
P [{\bf V}_{11}]_1[{\bf V}_{11}]_1^H$ as in (\ref{MEB}).
Therefore, the optimal ${\bf Q}_1$ at the boundary of the
achievable rate-energy region has a rank one.
\end{proof}

Note that, from (\ref{discuss3}), ${\bf Q}_2$ maximizing $\bar R$
is designed as
\begin{eqnarray}\label{discuss5}
{\bf Q}_2 = \arg \underset{{\bf Q}}{\max}~ tr({\bf H}_{22}{\bf
Q}{\bf H}_{22}^H)
\end{eqnarray}
and the corresponding ${\bf Q}_2$ is given by ${\bf Q}_2 = P [{\bf
V}_{22}]_1 [{\bf V}_{22}]_1^H$, where ${\bf V}_{22}$ is the right
singular matrix of ${\bf H}_{22}$. That is, at the low SNR, the
optimal information transfer strategy in the joint information and
energy transfer system is also a rank-one beamforming, which is
consistent with the result in the information transfer system
\cite{SAJafar}, where the region that the beamforming is optimal
becomes broader as the SNR decreases.

\subsection{Asymptotic behavior for a large $M$}\label{ssec:LargeM}
Note that Proposition \ref{prop2} gives us an insight on the joint
information and energy transfer with a large number of antennas
describing a promising future wireless communication structure
such as a massive MIMO system \cite{Marzetta, JHoydis,FRusek}.

Given the normalized channel ${\bf H}= \frac{\sqrt{M}}{\|\tilde
{\bf H}\|}\tilde {\bf H}$, where the elements of $\tilde{\bf H}$
are i.i.d. zero-mean complex Gaussian random variables (RVs) with
a unit variance, and ${\bf Q} = P{\bf v}{\bf v}^H$ with a finite
$P$ and $\|{\bf v}\|=1$, we define ${\bf R} = {\bf I}_M + {\bf H}
{\bf Q}{\bf H}^H$ which is analogous to ${\bf R}_{-2}$ in
(\ref{oneIDoneEH4}) of the proof of Proposition \ref{prop2}. Then,
because $\det({\bf R}) = 1+P {\bf v}^H{\bf H}^H{\bf H}{\bf v}$ and
\begin{eqnarray}\label{massiveMIMO1}
{\bf H}^H{\bf H} &\approx &\frac{1}{M}\tilde{\bf H}^H\tilde{\bf H}
~\approx~  {\bf I}_M, \quad (\text{Central limit theorem in
\cite{Papoulis}})\nonumber
\end{eqnarray}
when $M$ goes to infinity, $\det({\bf R})\approx 1+P$ and it is
independent from the beamforming vector ${\bf v}$. Analogously,
when $M$ increases, the design of ${\bf v}_1$ in ${\bf Q}_1=
P_1{\bf v}_1{\bf v}_1^H $ at the first transmitter is independent
from $\det({\bf R}_{-2})$ (accordingly, independent from ${\bf
H}_{21}$). Therefore, when nodes have a large number of antennas,
the transmit signal for energy transfer can be designed by caring
about its own link, not caring about the interference link to the
ID receiver. That is, for a large $M$, MEB with a power control
becomes optimal because it maximizes the energy transferred to its
own link.

\begin{remark}\label{remark3}
Interestingly, from Section \ref{ssec:LowSNRregime} and
\ref{ssec:LargeM} we note that, when the SNR decreases or the
number of antennas increases, the energy transfer strategy in the
MIMO IFC would be designed by only caring about its own link to
the EH receiver, not by considering the interference or leakage
through the other link to the ID receiver. In addition, massive
MIMO effect makes the joint information and energy transfer in the
MIMO IFC naturally split into disjoint information and energy
transfer in two non-interfering links.
\end{remark}

%
%
%
%
%

\section{Simulation Results}
\label{sec:simulation}

Computer simulations have been performed to evaluate the R-E
tradeoff of various transmission strategies in the two-user MIMO
IFC. In the simulations, the normalized channel ${\bf H}_{ij}$ is
generated such as ${\bf H}_{ij}=
\frac{\sqrt{\alpha_{ij}M}}{\|\tilde{\bf H}_{ij}\|_F}\tilde{\bf
H}_{ij}$, where the elements of $\tilde{\bf H}_{ij}$ are
independent and identically distributed (i.i.d.) zero-mean complex
Gaussian random variables (RVs) with a unit variance. The maximum
transmit power is set as $P=50W$, unless otherwise stated.

Table \ref{table_singlemode} lists the achievable rate and energy,
($R$, $E$), for single modes -- ($EH_1$, $EH_2$) and ($ID_1$,
$ID_2$). The harvested energy of ($EH_1$, $EH_2$) for $M=4$ is
larger than that for $M=2$ and furthermore, the achievable rate of
($ID_1$, $ID_2$) for $M=4$ is higher than that for $M=2$. Note
that the achievable rate of ($EH_1$, $EH_2$) and the harvested
energy of ($ID_1$, $ID_2$) are zero.

{\renewcommand\baselinestretch{1.2}{
\begin{table}
\renewcommand{\arraystretch}{1.2}
\caption{The achievable rate and energy, ($R$, $E$), for single
modes when $M \in \{2, 4\}$} \label{table_singlemode}
\begin{center}
\begin{tabular}{|c|c|c|}\hline
Mode & $M=2$  &  $M=4$ \\
\hline ($EH_1$, $EH_2$) &  (0, 262.98) & (0, 359.57) \\
\hline ($ID_1$, $ID_2$) &  (9.67, 0)  & (16.08, 0) \\
\hline
\end{tabular}
\end{center}
\end{table}
}}

 Fig. \ref{R_E_tradeoff_M4} shows R-E tradeoff curves for the
MEB and the MLB described in Section \ref{sec:rankoneBF} when $M =
4$, $\alpha_{ii} = 1$ for $i = 1,2$, and $\alpha_{ij} = 0.8$ for
$i \neq j$. The first transmitter takes a rank-one beamforming,
either MEB or MLB, and the second transmitter designs its transmit
signal as (\ref{oneIDoneEHmax4_1}), (\ref{oneIDoneEHmax5}), and
(\ref{oneIDoneEHmax6}), described in
Section \ref{ssec:REregion}. 
As expected, the MEB strategy raises the harvested energy at the
EH receiver, while the MLB increases the achievable rate at the ID
receiver. Interestingly, in the regions where the energy is less
than a certain threshold around 45 Joule/sec, the first
transmitter does not transmit any signals to reduce the
interference to the second ID receiver. That is, the energy
transferred from the second transmitter is sufficient to satisfy
the energy constraint at the EH receiver.

The dashed lines indicate the R-E curves of the time-sharing of
the full-power rank-one beamforming (either MEB or MLB) and the no
transmission at the first transmitter. Here, the second
transmitter switches between the beamforming on ${\bf H}_{12}$ as
(\ref{oneIDoneEHmax4_1}) and the water-filling as
(\ref{oneIDoneEHmax5}) in the corresponding time slots.
For MLB, ``water-filling-like'' approach (\ref{oneIDoneEHmax6})
exhibits higher R-E performance than the time-sharing scheme.
However, for MEB, when the energy is less than 120 Joule/sec, the
time-sharing exhibits better performance than the approach
(\ref{oneIDoneEHmax6}). That is, because the MEB causes large
interference to the ID receiver, it is desirable that, for the low
required harvested energy, the first transmitter turns off its
power in the time slots where the second transmitter is assigned
to exploit the
water-filling method 
as (\ref{oneIDoneEHmax5}). Instead, in the remaining time slots,
the first transmitter opts for a MEB with full power and the
second transmitter transfers its information to the ID receiver by
steering its beam on EH receiver's channel ${\bf H}_{12}$ as
(\ref{oneIDoneEHmax4_1}) to help the EH operation.
In Fig. \ref{R_E_tradeoff_M4_ranktwo}, we have additionally
included the R-E tradeoff curves for MEB with $rank({\bf Q}_1)=2$
when the simulation parameters are the same as those of Fig.
\ref{R_E_tradeoff_M4}. Here, we can find that the MEB with $rank
({\bf Q}_1)=1$ has superior R-E boundary points compared to that
with $rank ({\bf Q}_1)=2$. That is, even though we have not
identified the exact optimal R-E boundary, for a given beamforming
(MEB in Fig. 3), we can find that the beamforming with $rank ({\bf
Q}_1)=1$ has superior R-E boundary points compared to that with
$rank ({\bf Q}_1)=2$.

\begin{figure}
\begin{center}
\begin{tabular}{c}
\includegraphics[height=10cm]{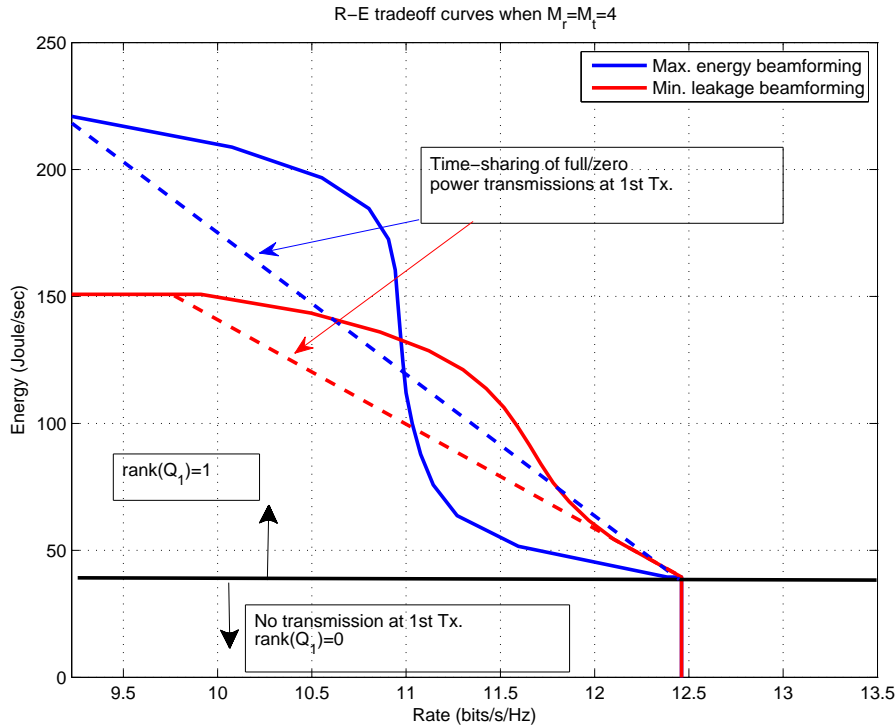} 
\end{tabular}
\end{center}
\caption[R_E_tradeoff]
{ \label{R_E_tradeoff_M4} R-E tradeoff curves for MEB and MLB when
$M = 4$, $\alpha_{ii} = 1$ for $i = 1,2$, and $\alpha_{ij} = 0.8$
for $i \neq j$.}
\end{figure}

\begin{figure}
\begin{center}
\begin{tabular}{c}
\includegraphics[height=10cm]{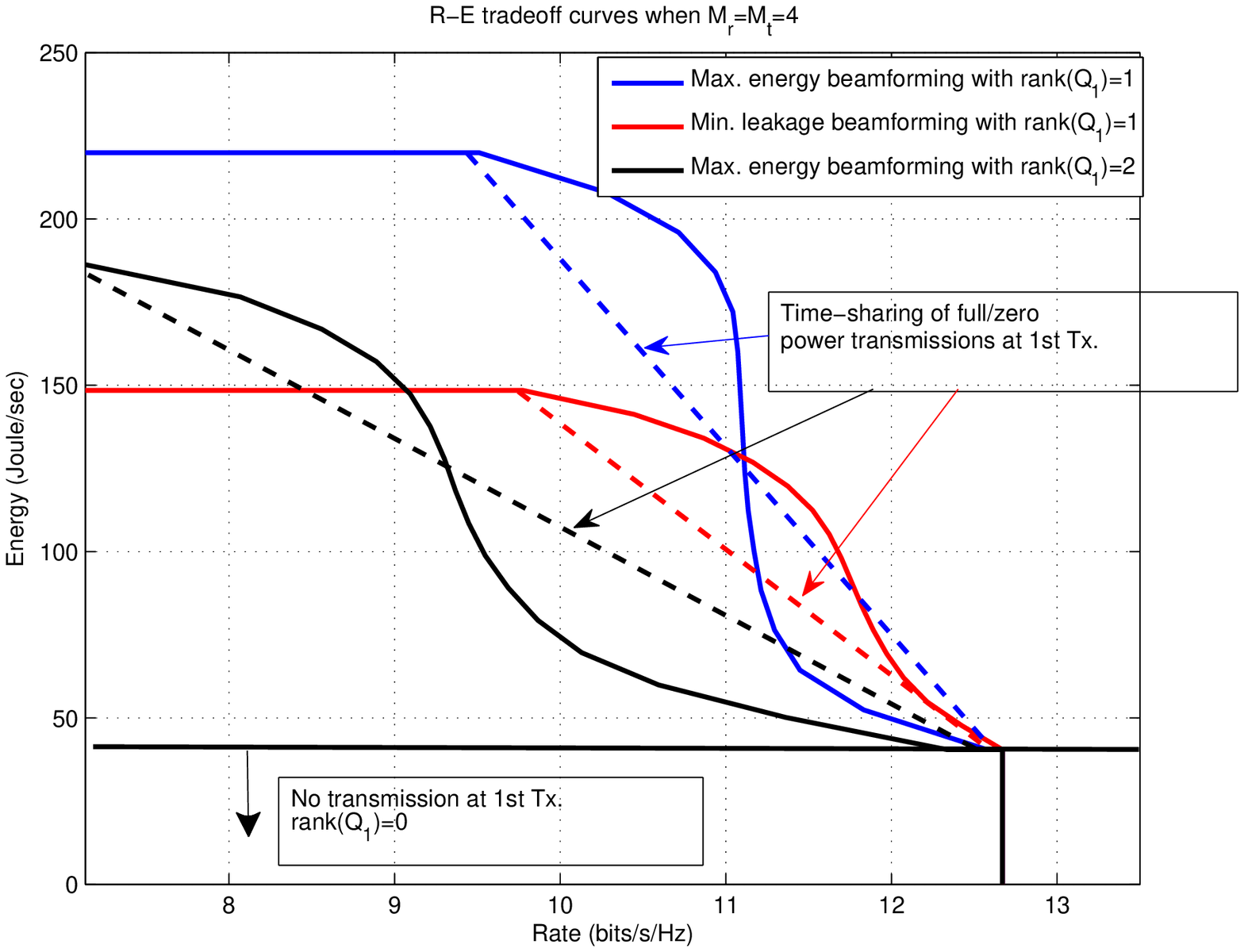} 
\end{tabular}
\end{center}
\caption[R_E_tradeoff]
{ \label{R_E_tradeoff_M4_ranktwo} R-E tradeoff curves for MEBs
($rank({\bf Q}_1) = \{1, 2\}$) and MLB ($rank({\bf Q}_1) = 1$)
when $M = 4$, $\alpha_{ii} = 1$ for $i = 1,2$, and $\alpha_{ij} =
0.8$ for $i \neq j$.}
\end{figure}

In Fig. \ref{R_E_tradeoff_M_4_lowSNR}, we plot R-E tradeoff curves
for $M=4$ and $P=0.1$. As observed in Section
\ref{ssec:LowSNRregime}, at the low SNR, the MEB exhibits higher
harvested energy than the MLB without any degradations in the
achievable rate. Fig. \ref{R_E_tradeoff_M15} shows R-E tradeoff
curves for $M = 15$. Compared to $M\in 4$ (Fig.
\ref{R_E_tradeoff_M4}), the gap between the achievable rates of
MEB and MLB is relatively less apparent. As pointed out in Remark
\ref{remark3} of Section \ref{sec:discussion}, for low SNRs or
large numbers of antennas in the MIMO IFC, the energy transfer
strategy of maximizing the transferred energy on its own link
exhibits wider R-E region than that of minimizing the interference
to the other ID receiver.

%
%

\begin{figure}
\begin{center}
\begin{tabular}{c}
\includegraphics[height=10cm]{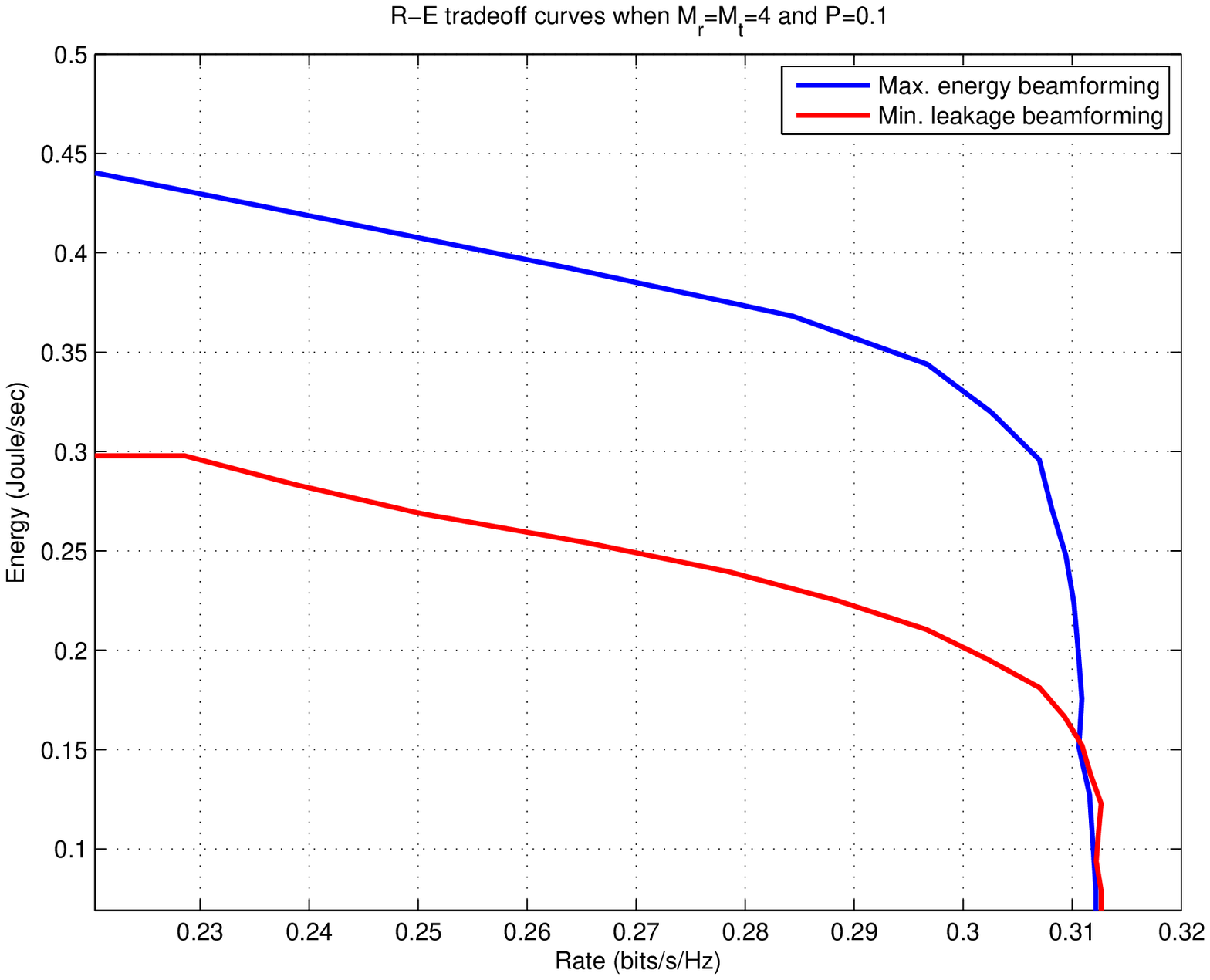} 
\end{tabular}
\end{center}
\caption[R_E_tradeoff]
{ \label{R_E_tradeoff_M_4_lowSNR} R-E tradeoff curves for MEB and
MLB when $M = 4$ and $P = 0.1$.}
\end{figure}

\begin{figure}
\begin{center}
\begin{tabular}{c}
\includegraphics[height=10cm]{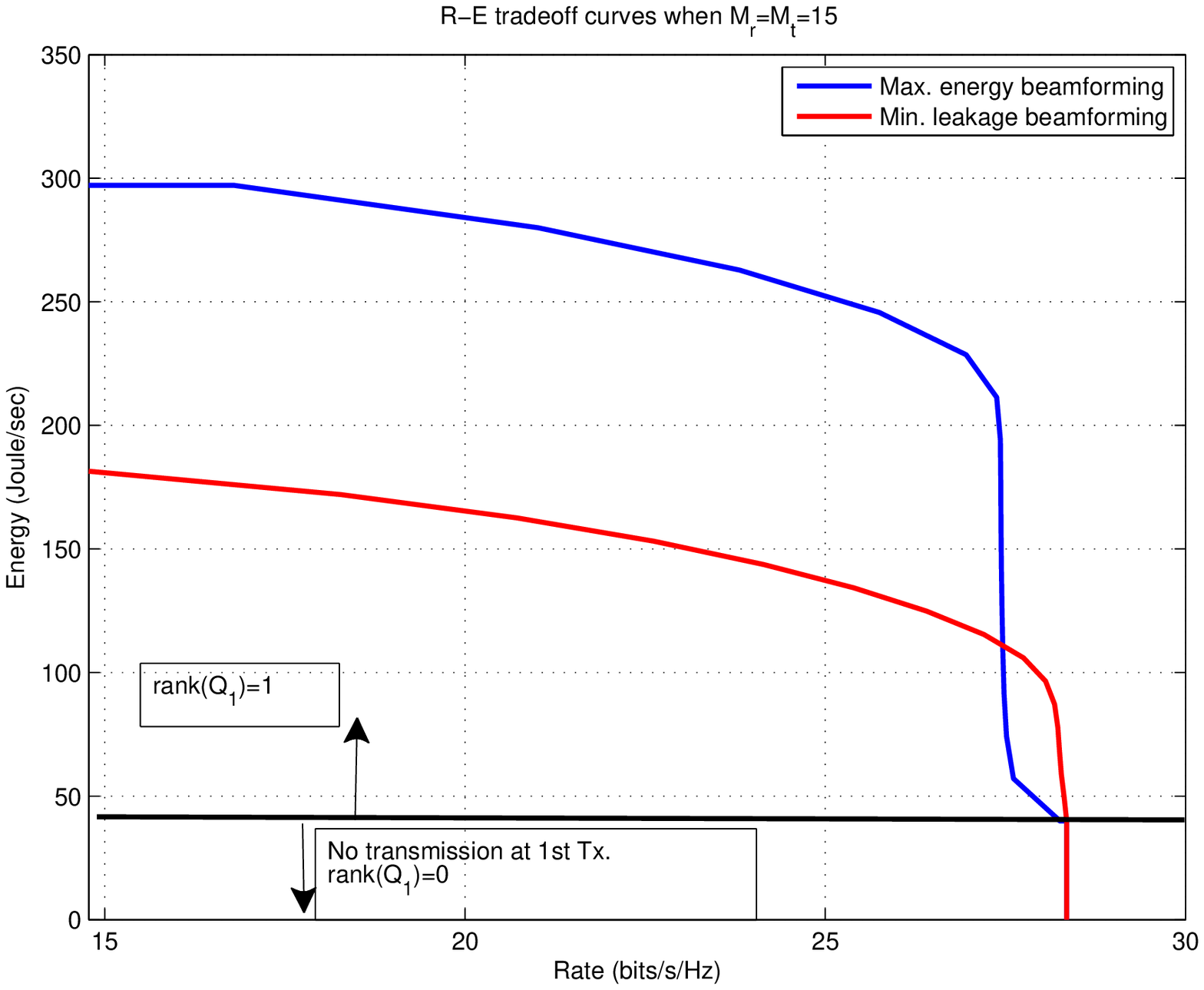} 
\end{tabular}
\end{center}
\caption[R_E_tradeoff]
{ \label{R_E_tradeoff_M15} R-E tradeoff curves for MEB and MLB
when $M = 15$.}
\end{figure}

 Fig. \ref{R_E_tradeoff_M_4_GSVD} shows R-E tradeoff curves for
MEB, MLB, SLNR maximizing beamforming, and SLER maximizing
beamforming when $M = 4$, $\alpha_{ii} = 1$ for $i = 1,2$, and
$\alpha_{ij} = 0.8$ for $i \neq j$. The R-E region of the proposed
SLER maximizing beamforming covers most of those of both MEB and
MLB, while the SLNR beamforming does not cover the region for MEB.
Fig. \ref{R_E_tradeoff_M_4_GSVD_asymmetric} shows R-E tradeoff
curves for an asymmetric case $M_t=3$ and $M_r= 4$. We can find a
similar trend with $M_t=4=M_r= 4$, but the overall harvested
energy with $M_t=3$ and $M_r= 4$ is slightly less than that with
$M_t=4=M_r= 4$.

\begin{figure}
\begin{center}
\begin{tabular}{c}
\includegraphics[height=10cm]{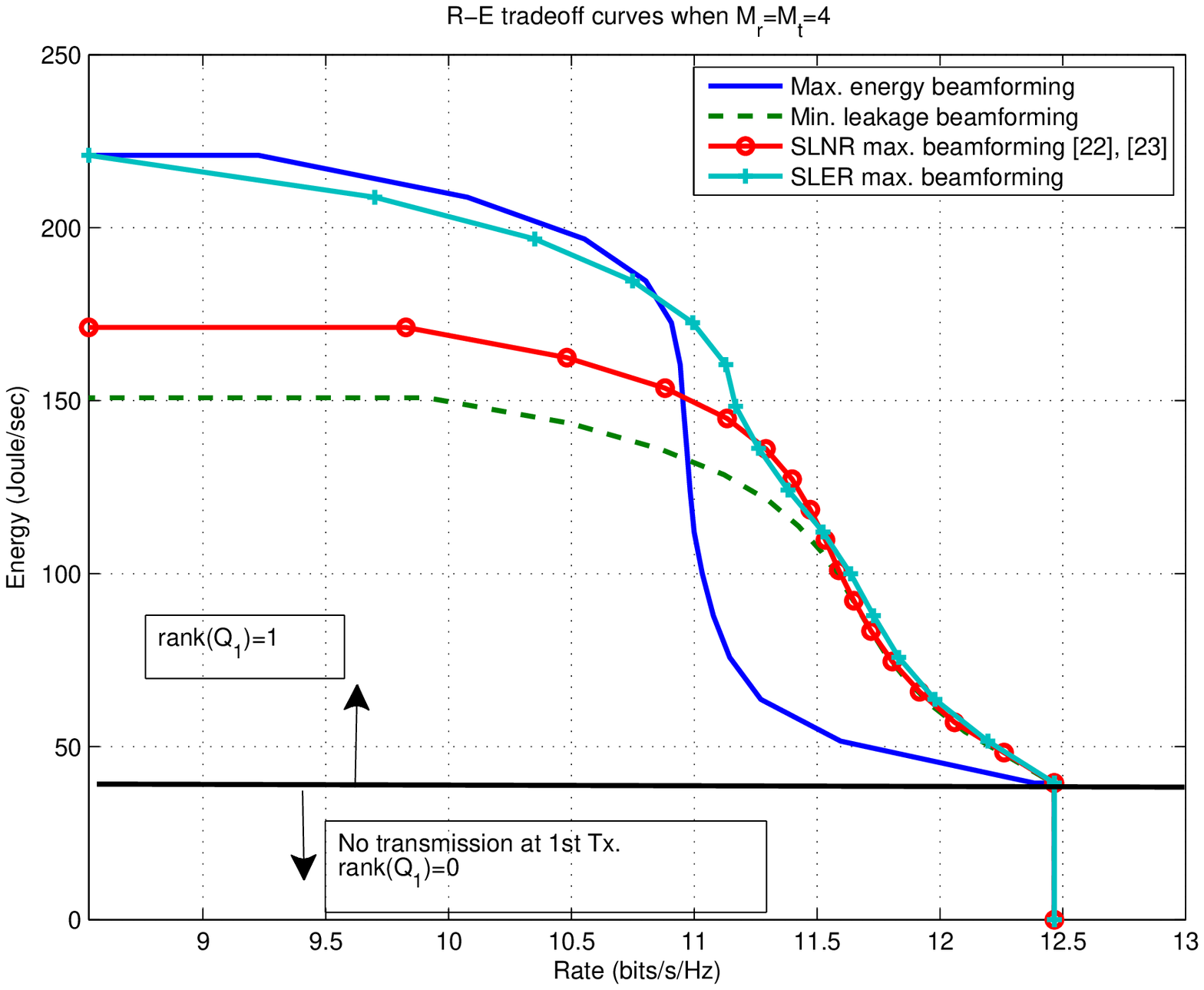}
\end{tabular}
\end{center}
\caption[R_E_tradeoff_M_4_GSVD]
{ \label{R_E_tradeoff_M_4_GSVD} R-E tradeoff curves for MEB, the
MLB, SLNR maximizing beamforming, and SLER maximizing beamforming
when $M = 4$, $\alpha_{ii} = 1$ for $i = 1,2$, and $\alpha_{ij} =
0.8$ for $i \neq j$.}
\end{figure}

\begin{figure}
\begin{center}
\begin{tabular}{c}
\includegraphics[height=10cm]{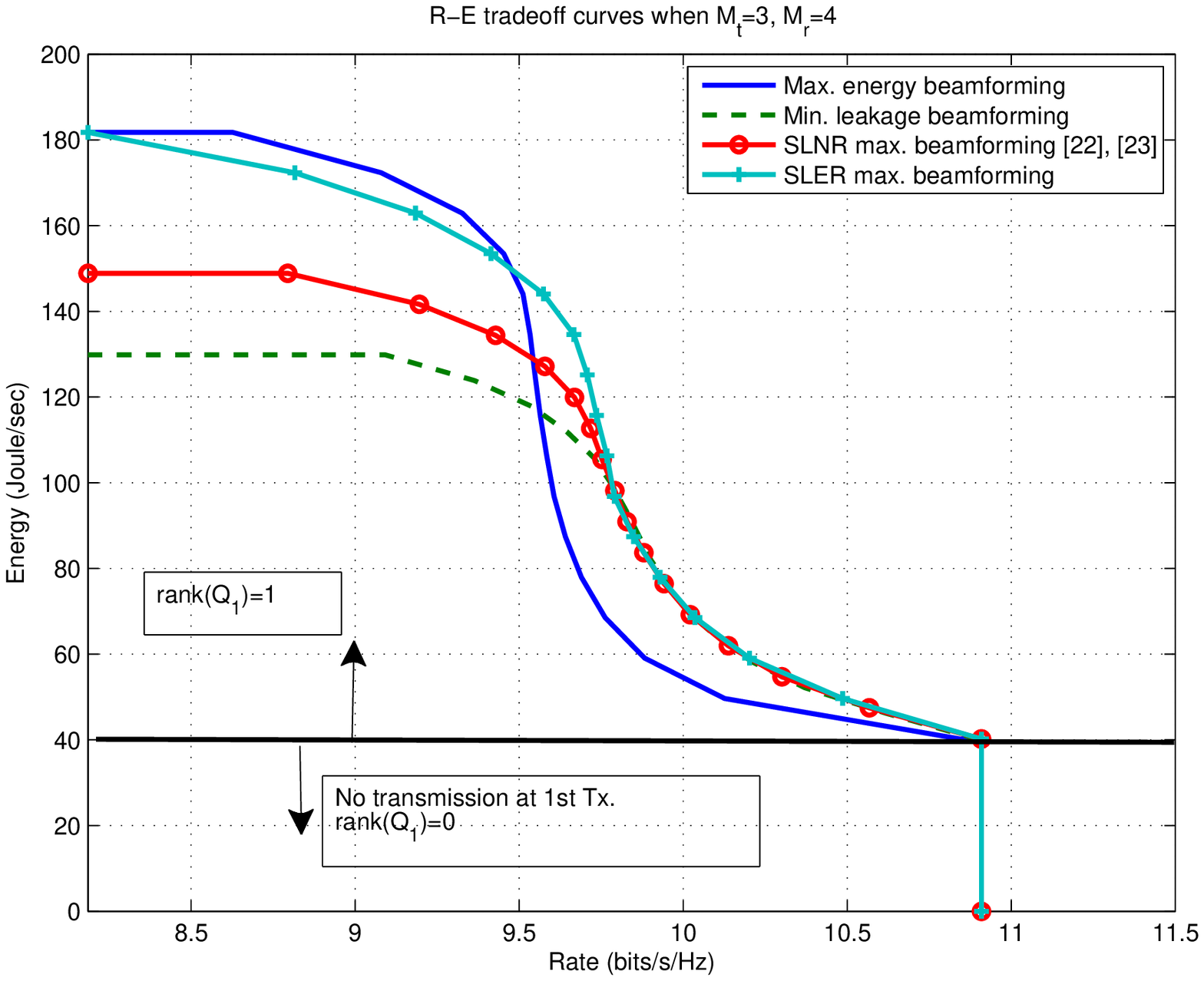}
\end{tabular}
\end{center}
\caption[R_E_tradeoff_M_4_GSVD]
{ \label{R_E_tradeoff_M_4_GSVD_asymmetric} R-E tradeoff curves for
MEB, the MLB, SLNR maximizing beamforming, and SLER maximizing
beamforming when $M_t = 3$ and $M_r = 4$.}
\end{figure}

Fig. \ref{R_E_tradeoff_M_2_GSVD_scheduling} shows the R-E tradeoff
curves for SLER maximizing beamforming with/without SLER-based
scheduling described in Section \ref{sec:SLER_maximizing} when (a)
$\alpha_{ij}= 0.7$ and (b) $\alpha_{ij}= 1$ for $i\neq j$. Here,
we set $\alpha_{ii} =1$ for $i = 1,2$ and $M = 2$. Note that the
case with $\alpha_{ij} = 0.7 $ has weaker cross-link channel
(inducing less interference) than that with $\alpha_{ij} = 1 $.
The SLER-based scheduling extends the achievable R-E region for
both $\alpha_{ij} \in \{0.7,1\}$, but the improvement for
$\alpha_{ij}=1$ is slightly more apparent. That is, the SLER-based
scheduling becomes more effective when strong interference exists
in the system. Note that the case with $\alpha_{ij} = 1 $ exhibits
slightly lower achievable rate than that with $\alpha_{ij}=0.7$,
while achieving larger harvested energy. That is, the strong
interference degrades the information decoding performance but it
can be effectively utilized in the energy-harvesting.

%

\begin{figure}
\centering 
 \subfigure[]
  {\includegraphics[height=6cm]{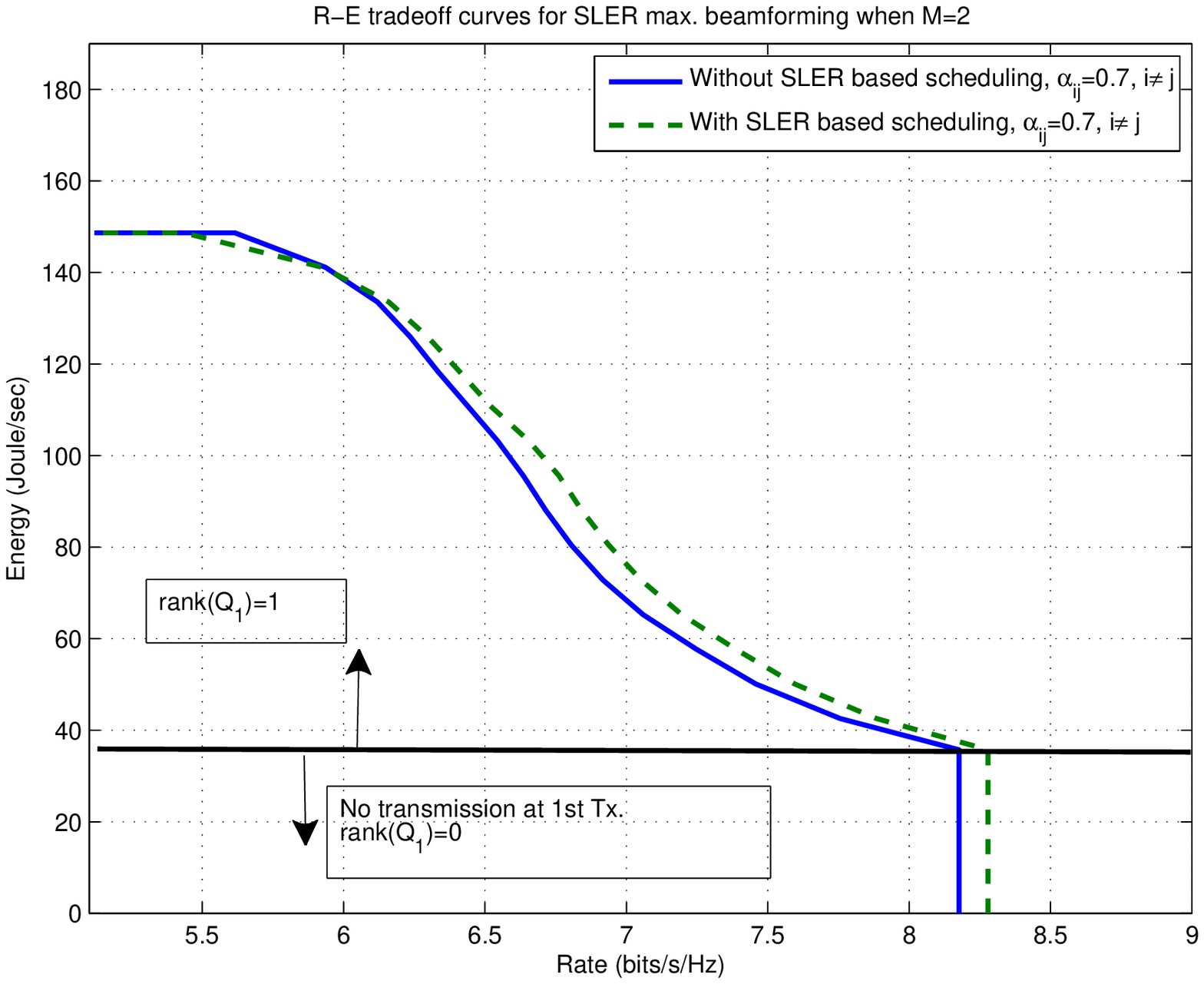}}\hspace{3em}
 \subfigure[]
  {\includegraphics[height=6cm]{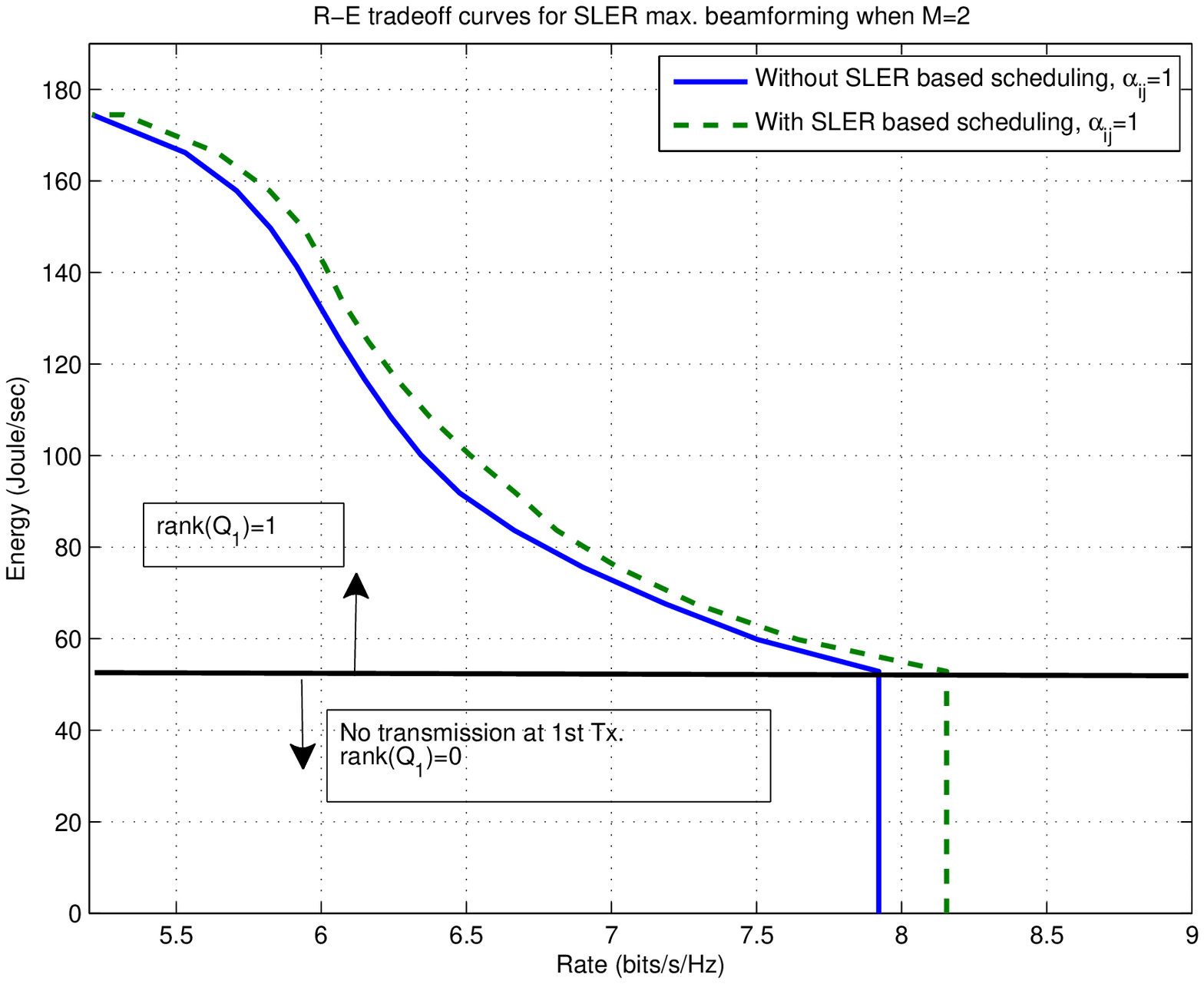}}
 \caption{   R-E tradeoff curves for
SLER maximizing beamforming with/without SLER-based scheduling
when (a) $\alpha_{ij}= 0.7$ and (b) $\alpha_{ij}= 1$ for $i\neq
j$. Here, $M=2$.} \label{R_E_tradeoff_M_2_GSVD_scheduling}
\end{figure}

\section{Conclusion}
\label{sec:conc} In this paper, we have investigated the joint
wireless information and energy transfer in two-user MIMO IFC.
Based on Rx mode, we have different transmission strategies. For
single-operation modes - ($ID_1$, $ID_2$) and ($EH_1$, $EH_2$),
the iterative water-filling and the energy-maximizing beamforming
on both receivers can be adopted to maximize the information bit
rate and the harvested energy, respectively. For ($EH_1$, $ID_2$),
and ($ID_1$, $EH_2$), we have found a necessary condition of the
optimal transmission strategy that one of transmitters should take
a rank-one beamforming with a power control. Accordingly, for two
transmission strategies that satisfy the necessary condition - MEB
and MLB, we have identified their achievable R-E tradeoff regions,
where the MEB (MLB) exhibits larger harvested energy (achievable
rate). We have also found that when the SNR decreases or the
number of antennas increases, the joint information and energy
transfer in the MIMO IFC can be naturally split into disjoint
information and energy transfer in two non-interfering links.
Finally, we have proposed a new transmission strategy satisfying
the necessary condition - signal-to-leakage-and-energy ratio
(SLER) maximization beamforming which shows wider R-E region than
the conventional transmission methods. That is, we have found that
even though the interference degrades the ID performance in the
two-user MIMO IFC, the proposed SLER maximization beamforming
scheme effectively utilizes it in the EH without compromising ID
performance.

Note that, motivated from the rank-one beamforming optimality, the
identification of the optimal R-E boundary will be a challenging
future work. Furthermore, the partial CSI or erroneous channel
information degrades the achievable rate and the harvested energy
at the receivers, which drives us to develop a robust rank-one
beamforming.

\section*{Acknowledgment}
The authors would like to thank the anonymous reviewers whose
comments helped us improve this paper.

\bibliographystyle{IEEEtran}
\bibliography{IEEEabrv,myref}

\end{document}